\documentclass[a4paper]{amsart}
\usepackage[utf8x]{inputenc}
\usepackage{txfonts}
\usepackage{amsmath,accents}
\usepackage{amsthm}
\usepackage{amsfonts}
\usepackage{amssymb}
\usepackage{array}
\usepackage{stmaryrd}
\usepackage{mathrsfs}
\usepackage[usenames]{color, colortbl}
\usepackage[sans]{dsfont}
\usepackage{graphicx}
\usepackage{tikz}
\usetikzlibrary{shapes.arrows,shapes.geometric,chains,decorations.pathreplacing,positioning,arrows,automata} 
\usepackage{wrapfig}

\newcommand{\elo}{Eloise}
\newcommand{\abe}{Abelard}
\newcommand{\G}{\mathds{G}}
\newcommand{\W}{\mathds{W}}
\newcommand{\T}{\mathscr{T}}
\renewcommand{\H}{H} 

\def\limp{\Rightarrow}
\def\liff{\Leftrightarrow}
\def\forallNat{\forall^{\mathsf{nat}}}
\def\existsNat{\exists^{\mathsf{nat}}}
\def\forallN{\forall^{\mathsf{N}}}
\def\existsN{\exists^{\mathsf{N}}}
\def\Lam#1#2{\lambda#1\,.\,#2}
\def\<{\langle}
\def\>{\rangle}
\newcommand{\tv}[1]{|#1|}
\newcommand{\fv}[1]{\|#1\|}
\newcommand{\far}[1]{\forallN #1}
\newcommand{\exr}[1]{\existsN {#1}}
\newcommand{\eval}{\succ}
\newcommand{\Th}{\mathop{\textbf{th}}}
\newcommand{\real}{\Vdash}
\newcommand{\ureal}{\Vvdash}
\newcommand{\emptystr}{\<\cdot\>}
\newcommand{\pole}{\Bot}
\newcommand{\Pl}{\ensuremath{\operatorname{PL}}}
\newcommand{\Eq}{\texttt{eq}}
\newcommand{\Quote}{\texttt{quote}}
\newcommand{\Next}{\texttt{next}}
\newcommand{\Fork}{{\pitchfork}}
\newcommand{\I}{\mathrm{\bf I}}
\renewcommand{\k}{\textrm{\bf k}}
\newcommand{\cc}{\textrm{\bf c\!c}}
\newcommand{\barr}{\overline}

\newcommand{\dom}{\mathop{\mathrm{dom}}}
\newcommand{\Int}[1]{\llbracket #1 \rrbracket}
\newcommand{\ds}{\displaystyle}
\newcommand{\FV}{\mathit{FV}}
\renewcommand{\P}{\mathcal{P}}

\newcommand{\TYP}[3]{#1~{\vdash}~#2~{:}~#3}

\newcommand{\Pow}{\mathfrak{P}}
\newcommand{\Nat}{\mathsf{Nat}}

\newcommand{\C}{\mathcal{C}}
\newcommand{\N}{\mathds{N}}

\newcommand{\V}{\mathcal{V}}
\newcommand{\B}{\mathcal{B}}
\newcommand{\A}{\mathcal{A}}
\newcommand{\M}{\mathcal{M}}

\newcommand{\pair}[1]{\langle #1 \rangle}
\newcommand{\union}[1]{\bigcup_{#1}}
\newcommand{\vc}[2]{\vec{#1}_{#2}}

{\newtheorem{thm}{Theorem}
\newtheorem{lem}{Lemma}
\newtheorem{prop}{Proposition}
\newtheorem{cor}{Corollary}
\newtheorem{clm}{Fact}

\theoremstyle{definition}
\newtheorem{definition}{Definition}
\newtheorem{example}{Example}
}

\usepackage{proof, bussproofs, soul}
\newcommand{\mauricio}[1]{#1}
\newcommand{\langue}[1]{}
\newcommand{\comment}[1]{}

\begin{document}
\title{Classical realizability and arithmetical formul\ae}
\author{Mauricio Guillermo}
\author{\'Etienne Miquey}
\address{(Mauricio Guillermo)\\ \textsc{Universidad de la Rep\'ublica\\ IMERL\\ Facultad de Ingenier\'ia\\ Montevideo\\ Uruguay}}
\email{mguille@fing.edu.uy}
\address{(\'Etienne Miquey)\\ \textsc{PPS Laboratory, Univ Paris Diderot, Team PiR2, INRIA\newline\indent Universidad de la Rep\'ublica\\ IMERL\\ Facultad de Ingenier\'ia\\ Montevideo\\ Uruguay}}
\email{etienne.miquey@ens-lyon.fr}

\begin{abstract}
In this paper we treat the specification problem in Krivine classical realizability~\cite{Kri09},
in the case of arithmetical formul\ae.
In the continuity of previous works from Miquel and the first author~\cite{GuiPhD,Peirce}, we characterize the universal realizers of a formula 
as being the winning strategies for a game (defined according to the formula).
In the first section we recall the definition of classical realizability, as well as a few technical results.
In Section~\ref{s:specification}, we introduce in more details the specification problem 
and the intuition of the game-theoretic point of view we adopt later.
We first present a game $\G^{1}$, that we prove to be adequate and complete if the language contains no instructions `\texttt{quote}'~\cite{Kri03},
using interaction constants to do substitution over execution threads. 
We then show that as soon as the language contain `\texttt{quote}', the game is no more complete, and present a second game 
$\G^{2}$ that is both adequate and complete in the general case.
In the last Section, we draw attention to a model-theoretic point of view and 
use our specification result to show that arithmetical formul\ae\ are absolute for realizability models.
\end{abstract}
\maketitle

\section{Introduction}
The so called \emph{Curry-Howard correspondence} constituted an important breakthrough in proof theory,
by evidencing a strong connection between the notions of functional programming and proof theory \cite{CF58,How69,Gir89}. 
For a long time, this correspondence has been limited to intuitionistic proofs and constructive mathematics,
so that classical reasonings, that are omnipresent in mathematics, could only be retrieved through 
negative translations to intuitionistic logic~\cite{Fri78} or to linear
logic~\cite{Gir06}. 

In 1990, Griffin discovered that the control operator \texttt{call/cc} (for \emph{call with current continuation})
of the Scheme programming language could be typed by Peirce's law $((A\to B)\to A)\to A)$,
this way extending the formu\ae-as-types interpretation \cite{How69}.
As Peirce's law is known to imply, in an intuitionistic framework, all the other forms of classical reasoning 
(excluded middle, \emph{reductio ad absurdum}, double negation elimination, etc.),
this discovery opened the way for a direct computational interpretation of classical proofs,
using control operators and their ability to \emph{backtrack}. 
Several calculi were born from this idea, such as Parigot's $\lambda\mu$-calculus~\cite{Par97}, 
Barbanera
and Berardi's symmetric $\lambda$-calculus~\cite{BB96}, 
Krivine's
$\lambda_c$-calculus~\cite{Kri09} or Curien and Herbelin's
$\bar{\lambda}\mu\tilde\mu$-calculus~\cite{CH00}.

Nonetheless, some difficulties quickly appeared in the analysis of the computational behaviour of programs extracted from classical proofs.
One reason for these difficulties was precisely the presence of control operators, 
whose ability to backtrack breaks the linearity of the execution of programs. 
More importantly, the formul{\ae}-as-types interpretation suffered from the lack of a theory 
connecting the point of view of typing with the point of view of computation.
Realizability was designed by Kleene to interpret the computational contents of
the proofs of Heyting arithmetic~\cite{Kle45}, and even if it has been extended later to more general frameworks 
(like intuitionistic set theories~\cite{Myh73,Fri73,McCPhD}),
it is intrinsically incompatible with classical reasoning:
the negation of the middle excluded principle is realizable.

\subsection{Classical realizibility}
To address this problem, Krivine introduced in the middle of the 90s the theory of \emph{classical realizability}~\cite{Kri09}, 
which is a complete reformulation of the very principles of realizability to make them compatible with classical reasoning. 
(As noticed in several articles~\cite{Oli08,Miq10}, classical realizability can be seen as a reformulation of Kleene's realizability 
through Friedman's $A$-translation~\cite{Fri78}.) 
Although it was initially introduced to interpret the proofs of classical second-order arithmetic, 
the theory of classical realizability can be scaled to more expressive theories such as Zermelo-Fraenkel 
set theory~\cite{Kri01} or the calculus of constructions with universes~\cite{Miq07}.  

As in intuitionistic realizability, every formula~$A$ is interpreted in classical realizability as a set~$|A|$ of programs 
called the \emph{realizers} of~$A$, that share a common computational behaviour dictated by the structure of the formula~$A$. 
This point of view is related to the point of view of deduction (and of typing) via the property of \emph{adequacy}, 
that expresses that any program extracted from a proof of~$A$---that is: any program of type~$A$---realizes the formula~$A$, 
and thus has the computational behaviour expected from the formula~$A$.  

However the difference between intuitionistic and classical
realizability is that in the latter,  
the set of realizers of~$A$ is defined indirectly, that is from a set
$\|A\|$ of execution contexts (represented as argument stacks)  
that are intended to challenge the truth of~$A$. 
Intuitively, the set~$\|A\|$---which we shall call the \emph{falsity
  value of~$A$}---can be understood as the set of all possible  
counter-arguments to the formula~$A$.  
In this framework, a program realizes the formula~$A$---i.e.\ belongs
to the \emph{truth value}~$|A|$---if and only if  
it is able to defeat all the attempts to refute~$A$ using a stack in~$\|A\|$. 
(The definition of the classical notion of a realizer is also
parameterized by a \emph{pole} representing a particular challenge,  
that we shall define and discuss in Section~\ref{sss:Poles}.)  

By giving an equal importance to programs---or terms---that `defend'
the formula~$A$, and to execution contexts---or stacks---that `attack'  
the formula~$A$, the theory of classical realizability is therefore
able to describe the interpretation of classical reasoning in terms of  
manipulation of whole stacks (as first class citizens) using control operators.

\subsection{Krivine $\lambda_c$-calculus}

The programming language commonly used in classical realizability is
Krivine's $\lambda_c$-calculus, which is an extension of Church's
$\lambda$-calculus~\cite{Chu41}
containing an instruction $\cc$ (representing the control operator
\texttt{call/cc}). 
and \emph{continuation constants} embedding stacks.
Unlike the traditional $\lambda$-calculus, the $\lambda_c$-calculus 
is parameterized by a particular execution strategy ---corresponding 
to the Krivine Abstract Machine~\cite{Kam}---
so that the notion of confluence---which is central in traditional
$\lambda$-calculi, does not make sense anymore. 
The property of confluence is replaced by the property of determinism,
which is closer from the point of view of real programming languages. 

A pleasant feature of this calculus is that 
it can be enriched with \emph{ad hoc} extra instructions.
For instance, a \texttt{print} instruction might be added to trace an execution,
as well as extra instructions manipulating primitive numerals to do
some code optimization~\cite{Miq10}. 
In some situations, extra instructions can also be designed to realize
reasoning principles, 
the standard example being the instruction $\Quote$ that computes the
Gödel code of a stack, 
used for instance to realize the axiom of dependent choice~\cite{Kri03}.
In this paper, we shall consider this instruction together $\Eq$, 
that tests the syntactic equality between two $\lambda_c$-terms.

\subsection{The specification problem}
A central problem in classical realizability is the
\emph{specification problem}, 
which is to find a characterization for the (universal) realizers of a formula
by their computational behaviour.
In intuitionistic logic, this characterization does not contain more
information than  
the formula itself, so that this problem has been given little attention. 
For instance, the realizers of an existential formula $\existsN x A(x)$
are exactly the ones reducing to a pair made of a witness $n\in\N$ and
a proof term realizing $A(n)$~\cite{Kri93}. 

However, in classical realizability the situation appears to be quite different 
and the desired characterization is in general much more difficult to obtain.
Indeed, owing to the presence of control operators in the language of terms,
the realizers have the ability to backtrack at any time, making the
execution harder to predict. 
Considering for instance the very same formula $\existsN x A(x)$,
a classical realizer of it can give as many integers for $x$ as it wants, 
using backtrack to make another try.
Hence we can not expect from such a realizer to reduct directly to a witness
(for an account of witness extraction techniques in classical
realizability, see Miquel's article~\cite{Miq10}). 
In addition, as we will see in Section~\ref{ss:halt}, 
giving such a witness might be computationally impossible without backtrack, 
for example in the case of a formula relying on the Halting Problem.
We will treat this particular example in Section~\ref{ss:halt}.

Furthermore, as stated in the article on Peirce's Law~\cite{Peirce}, the presence of instructions such as $\Quote$ makes the problem
still more subtle. We will deal with this particular case in Section~\ref{s:nsubst}.

\subsection{Specifying arithmetical formul\ae}
The architecture of classical realizability is centered around the
opposition between falsity values (stacks) 
and truth values (terms). This opposition, as well as the underlying
intuition (opponents vs. defenders), 
naturally leads us to consider the problem in a game-theoretic setting.
Such a setting---namely realizability games--- was defined by Krivine
as a mean to \mauricio{prove that any arithmetical formula which
  is \emph{universally realized} (i.e.:
  realized for all poles) is true in the ground model
  (often considered as the 
  \emph{standard full model of second order 
  arithmetics})~(c.f.: theorem 16 of \cite{Kri03}). Thereafter, Krivine also
  proves the converse, which is that every arithmetical formula which
  is true in the ground model is realized by the term 
  implementing the trivial winning strategy of the game associated to
  the formula (c.f.: theorem 21 of
  \cite{Kri09}). These realizability games are 
  largely inspired on the \emph{non-counterexample interpretation} of
  Kreisel \cite{Kreisel51}, \cite{Kreisel52} and the subsequent
  developpement of game semantics for proofs by
  Coquand~\cite{Coquand95}.} 

  \mauricio{Our goal is to establish an operational
  description which characterize \emph{all} the realizers of a
  given arithmetical formula. In particular, 
  it does not suffice to find a realizer for any true arithmetical 
  formula, but we want to explicit a sufficient operational condition to
  be a realizer.}  

  In Coquand's games, the only atomic
  formul\ae~ are $\top$ and $\bot$, 
  therefore a strategy   
  for a true atomic formula does nothing, as the game is already won by
  the defender.  
  As a consequence, any ``blind'' enumeration of $\N^k$ is a winning
  strategy for every true $\Sigma^0_{2k}$-formul\ae. 
  Such a strategy, which is central in Krivine's proof that any true
  formula in the ground model is realized~\cite[Theorem 21]{Kri09}, 
  has no interesting computational content. \mauricio{Even more}, it is
  not suitable 
  for being a realizer in the general case where we use Leibniz
  equality.  
  This remark will be discussed more consistently in Section \ref{s:models}. 

  \mauricio{The game developped by Krivine makes
  both players to use only constants. If the calculus does not contain
  instructions incompatible with substitution (like {\tt 'quote'}),
  this game is   
  equivalent to the one we prove that specifies the arithmetical
  formul\ae\ in the substitutive case. However, Krivine's realizers are
  eventually intended to contain {\tt 'quote'}. In this general case,
  we prove that the specification is obtained from the first game by a
  relaxation of the rules of $\exists$.}
 

  \mauricio{Thus,} both works left open the question of giving a
    precise specification for arithmetic formul\ae~ in the general case.

In this paper we will rephrase the game-theoretic framework of the
first author Ph.D. thesis~\cite{GuiPhD}  
to provide a game-theoretic characterization $\G^{1}$ that is both complete and adequate, 
in the particular case where the underlying calculus contains
infinitely many \emph{interaction constants}. 
However, this hypothesis---that is crucial in our proof of completeness---is known to be incompatible 
with the presence of instructions such as $\Quote$ or $\Eq$~\cite{Peirce}, which allow us to distinguish syntactically $\lambda_c$ terms 
that are computationally equivalent. 
We exhibit in Section~\ref{ss:wild} a \emph{wild realizer} 
that uses these instructions and does not suit as a winning strategy for $\G^{1}$, 
proving that $\G^{1}$ is no more complete in this case.

Indeed, as highlighted in the article on Peirce's Law~\cite{Peirce}, the presence of such instructions 
introduces a new---and purely game-theoretic---form of backtrack that does not come from a control operator
but from the fact that realizers, using a syntactic equality test provided by $\Quote$, 
can check whether a position has already appeared before.
We present in Section~\ref{s:nsubst} a second game $\G^{2}$ that allows this new form of backtrack, and captures the behaviour 
of our wild realizer. Then we prove that without any assumption on the set of instructions,
this game is both adequate and complete, thus constituting the definitive specification of arithmetical formul\ae.

\subsection{Connexion with forcing}
\label{ss:introforcing}
In addition to the question of knowing how to specify arithmetical formul\ae,
this paper presents an answer to another question, which is to know 
whether arithmetical formul\ae\ are \emph{absolute} for realizability models.
In set theory, a common technique to prove independence results in theory is to use forcing,
that allows us to extend a model and add some specific properties to it.
Yet, it is known $\Sigma^{1}_2$-formul\ae~are absolute for a large class of models, including those produced by forcing.
This constitutes somehow a barrier to forcing, 
which  does not permit to change the truth of formul\ae~that are below $\Sigma^{1}_2$ in the arithmetical hierarchy.

If classical realizability was initially designed to be a semantics for proofs of Peano second-order arithmetic,
it appeared then to be scalable to build models for high-order arithmetic~\cite{Forcing} or set theory~\cite{RAlgI}.
Just like forcing techniques, these constructions rest upon a ground model and allow us to break some formul\ae~ 
that were true in the ground model, say the continuum hypothesis or the axiom of choice~\cite{Gimmel}.
In addition, the absoluteness theorem of $\Sigma^{2}_1$ does not apply to realizability model.
Hence it seems quite natural to wonder, as for forcing, whether realizability models preserve some formul\ae.
We will explain in Section \ref{s:models} how the specification results allow us to show that arithmetical formul\ae\ 
are absolute for realizability models.

\section{The language $\lambda_c$}
\label{s:Lamc}

A lot of the notions we use in this paper are the very same as in the article on Peirce's Law~\cite{Peirce}. 
We will recall them briefly, for a more gentle introduction, 
we advise the reader to refer to this paper.
\newcommand{\tabspace}{\\[-1.5ex]}
\newcommand{\leg}[1]{\renewcommand{\arraystretch}{1.5}\setlength{\tabcolsep}{1cm}
		      \begin{tabular}{c}
		          \textbf{#1}
		      \end{tabular}\\}


\begin{figure}[ht]
\begin{tabular}{|c|}
\hline
\tabspace
 \begin{tabular}{>{\bf}l@{\hspace{2cm}}ccl@{\hspace{1.5cm}}r}
    Terms & $t,u$ &$::=$& $x  \mid \lambda x . t \mid tu \mid  \k_\pi \mid \kappa$ & $x,\in\V_\lambda,\kappa\in C$ \\
    Stacks & $\pi$ &$::=$& $\alpha ~~|~~ t\cdot \pi $ & ($\alpha\in\B$,~$t$ closed)\\ 
    Processes & $p, q$ &$::=$& $t\star \pi$ & ($t$ closed)\\
  \end{tabular}
\\
\tabspace
\hline
\tabspace
\small 
$\begin{array}{rcl}
  x\{c:=u\} &\equiv& x \\
  (\Lam{x}{t})\{c:=u\} &\equiv& \Lam{x}{t\{c:=u\}} \\
  (t_1t_2)\{c:=u\} &\equiv& t_1\{c:=u\}t_2\{c:=u\} \\
  \k_{\pi}\{c:=u\} &\equiv& \k_{\pi\{c:=u\}} \\
  c\{c:=u\} &\equiv& u \\
  c'\{c:=u\} &\equiv& c'\hfill(\text{if}~c'\not\equiv c) \\
  \alpha\{c:=u\} &\equiv& \alpha \\
  (t\cdot\pi)\{c:=u\} &\equiv& t\{c:=u\}\cdot\pi\{c:=u\} \\
\end{array}\quad
\begin{array}{@{}rcl}
  x\{\alpha:=\pi_0\} &\equiv& x \\
  (\Lam{x}{t})\{\alpha:=\pi_0\} &\equiv&
  \Lam{x}{t\{\alpha:=\pi_0\}} \\
  (t_1t_2)\{\alpha:=\pi_0\} &\equiv&
  t_1\{\alpha:=\pi_0\}t_2\{\alpha:=\pi_0\} \\
  \k_{\pi}\{\alpha:=\pi_0\} &\equiv& \k_{\pi\{\alpha:=\pi_0\}} \\
  c\{\alpha:=\pi_0\} &\equiv& c \\
  \alpha\{\alpha:=\pi_0\} &\equiv& \pi_0 \\
  \alpha'\{\alpha:=\pi_0\} &\equiv& \alpha'\hfill
  (\text{if}~\alpha'\not\equiv\alpha) \\
  (t\cdot\pi)\{\alpha:=\pi_0\} &\equiv&
  t\{\alpha:=\pi_0\}\cdot\pi\{\alpha:=\pi_0\} \\
\end{array}$
 \\
 \normalsize
\leg{Substitution over terms and stacks}
\tabspace
\hline
\tabspace
\begin{tabular}{>{\bf}lc@{$::=$}lr}
First-order terms& $e_1,e_2$& $x\mid f(e_1,\ldots,e_k)$ & $x\in\V_1,f\in\Sigma$\\
Formul\ae & $A,B$ & $X(e_1,\ldots,e_k)\mid A\limp B \mid \forall x A\mid \forall X A$ & $X\in\V_2$
\end{tabular}
\\ 
\tabspace
\hline
\tabspace
\small
 $\begin{array}{@{}rcl}
  \bot &\equiv& \forall Z\,Z \\
  \lnot A &\equiv& A\limp\bot \\
  A\land B &\equiv& \forall Z\,((A\limp B\limp Z)\limp Z) \\
  A\lor B &\equiv& \forall Z\,((A\limp Z)\limp(B\limp Z)\limp Z) \\
\end{array}\begin{array}{rcl@{}}
  A\liff B &\equiv& (A\limp B)\land(B\limp A) \\
  \exists x\,A(x) &\equiv& \forall Z\,
  (\forall x\,(A(x)\limp Z)\limp Z) \\
  \exists X\,A(X) &\equiv& \forall Z\,
  (\forall X\,(A(X)\limp Z)\limp Z) \\
  e_1=e_2 &\equiv& \forall W\,(W(e_1)\limp W(e_2)) \\
\end{array}$\\ 
\normalsize
\leg{Second-order encodings} 
\tabspace
\hline
\tabspace
\renewcommand{\arraystretch}{2.5}
\setlength{\tabcolsep}{1cm}
$\begin{array}{ccc}

 \infer[\scriptstyle(x:A)\in\Gamma]{\TYP{\Gamma}{x}{A}}{} &
 
    \infer{\TYP{\Gamma}{\Lam{x}{t}}{A\limp B}}{
      \TYP{\Gamma,x:A}{t}{B}
    } 
    &
    \infer{\TYP{\Gamma}{tu}{B}}{
      \TYP{\Gamma}{t}{A\limp B} &\quad \TYP{\Gamma}{t}{A}
    } \\
    
    \infer[\scriptstyle x\notin\FV(\Gamma)]
    {\TYP{\Gamma}{t}{\forall x\,A}}{\TYP{\Gamma}{t}{A}}
    &
    \infer{\TYP{\Gamma}{t}{A\{x:=e\}}}{
      \TYP{\Gamma}{t}{\forall x\,A}
    } 
    &
    \infer[\scriptstyle X\notin\FV(\Gamma)]
    {\TYP{\Gamma}{t}{\forall X\,A}}{\TYP{\Gamma}{t}{A}}
    \\
    \infer{\TYP{\Gamma}{t}{A\{X:=P\}}}{
      \TYP{\Gamma}{t}{\forall X\,A}
    } 
    &
    \multicolumn{2}{c}{\infer{\TYP{\Gamma}{\cc}{((A\limp B)\limp A)\limp A}}{}}
 \end{array}$ \\
 \leg{Typing rules of second-order logic}[-1.5ex]\\
    
    \hline
\end{tabular}
\caption{Definitions}
\label{fig:bnf}
\end{figure}

\subsection{Terms and stacks}
\label{ss:TermsStacks}

The $\lambda_c$-calculus distinguishes two kinds of syntactic
expressions: \emph{terms}, which represent programs, and \emph{stacks},
which represent evaluation contexts.
Formally, terms and stacks of the $\lambda_c$-calculus are 
defined (see Fig.~\ref{fig:bnf}) from three auxiliary sets of symbols, that are pairwise
disjoint:
\begin{itemize}
\item A denumerable set~$\V_{\lambda}$ of $\lambda$-variables
  (notation: $x$, $y$, $z$, etc.)
\item A countable set~$\C$ of instructions, which contains at least an
  instruction $\cc$ (`call$/$cc', for: \emph{call with current
    continuation}).
\item A nonempty countable set~$\B$ of stack constants, also called
  stack bottoms (notation: $\alpha$, $\beta$, $\gamma$, etc.)
\end{itemize}
%


In what follows, we adopt the same writing conventions as in the pure
$\lambda$-calculus, by considering that application is
left-associative and has higher precedence than abstraction.
We also allow several abstractions to be regrouped under a
single~$\lambda$, so that the closed term
$\Lam{x}{\Lam{y}{\Lam{z}{((zx)y)}}}$ can be more simply written
$\Lam{xyz}{zxy}$.

As usual, terms and stacks are considered up to
$\alpha$-conversion~\cite{Bar84} and we denote by $t\{x:=u\}$ the
term obtained by replacing every free occurrence of the variable~$x$
by the term~$u$ in the term~$t$, possibly renaming the bound variables
of~$t$ to prevent name clashes.
The sets of all closed terms and of all (closed) stacks are
respectively denoted by~$\Lambda$ and~$\Pi$.

\begin{definition}[Proof-like terms]
  -- We say that a $\lambda_c$-term~$t$ is \emph{proof-like} if~$t$
  contains no continuation constant~$\k_{\pi}$. We denote by $\Pl$ the
  set of all proof-like terms. 
\end{definition}


Finally, every natural number $n\in\N$ is represented in the
$\lambda_c$-calculus as the closed proof-like term~$\barr{n}$ defined
by
$$\barr{n}~\equiv~\barr{s}^n\barr{0}~\equiv~
\underbrace{\barr{s}(\cdots(\barr{s}}_n\barr{0})\cdots)\,,$$
where $\barr{0}\equiv\Lam{xf}{x}$ and
$\barr{s}\equiv\Lam{nxf}{f(nxf)}$ are Church's encodings of zero and
the successor function in the pure $\lambda$-calculus.
Note that this encoding slightly differs from the traditional encoding
of numerals in the $\lambda$-calculus, although the term
$\barr{n}\equiv\barr{s}^n\barr{0}$ is clearly $\beta$-convertible to
Church's encoding $\Lam{xf}{f^nx}$---and thus computationally
equivalent.
The reason for preferring this modified encoding is that it is better
suited to the call-by-name discipline of Krivine's Abstract Machine
(KAM) we will now present.

\subsection{Krivine's Abstract Machine}

In the $\lambda_c$-calculus, computation occurs through the
interaction between a closed term and a stack within Krivine's Abstract
Machine (KAM).
Formally, we call a \emph{process} any pair $t\star\pi$ formed by a
closed term~$t$ and a stack~$\pi$.
The set of all processes is written $\Lambda\star\Pi$ (which is just
another notation for the Cartesian product of~$\Lambda$ by~$\Pi$).

\begin{definition}[Relation of evaluation]
 We call a relation of \emph{one step evaluation} any binary
  relation $\eval_1$ over the set $\Lambda\star\Pi$ of processes that
  fulfils the following four axioms: 
  $$\begin{array}{r@{~}c@{~}lcr@{~}c@{~}l}
    tu&\star&\pi &\eval_1& t&\star&u\cdot\pi \\
    (\Lam{x}{t})&\star&u\cdot\pi &\eval_1& t\{x:=u\}&\star&\pi \\
    \cc&\star&t\cdot\pi &\eval_1& t&\star&\k_{\pi}\cdot\pi \\
    \k_{\pi}&\star&t\cdot\pi' &\eval_1& t&\star&\pi \\
  \end{array}\leqno\begin{array}{@{}l}
  \textsc{(Push)}\\\textsc{(Grab)}\\
  \textsc{(Save)}\\\textsc{(Restore)}\\
  \end{array}$$
  The reflexive-transitive closure of~$\eval_1$ is written $\eval$.
\end{definition}

One of the specificities of the $\lambda_c$-calculus is that it comes
with a binary relation of (one step) evaluation~$\eval_1$ that is not
\emph{defined}, but \emph{axiomatized} via the rules (\textsc{Push}),
(\textsc{Grab}), (\textsc{Save}) and (\textsc{Restore}).
In practice, the binary relation~$\eval_1$ is simply another parameter
of the definition of the calculus, just like the sets~$\C$ and~$\B$.
Strictly speaking, the $\lambda_c$-calculus is not a particular
extension of the $\lambda$-calculus, but a family of extensions of the
$\lambda$-calculus parameterized by the sets~$\B$, $\C$ and the
relation of one step evaluation~$\eval_1$.
(The set~$\V_{\lambda}$ of $\lambda$-variables---that is
interchangeable with any other denumerable set of symbols---does not
really constitute a parameter of the calculus.)

\subsection{Adding new instructions}
\label{ss:ExtraInstr}

The main interest of keeping open the definition of the sets~$\B$,
$\C$ and of the relation evaluation~$\eval_1$ (by axiomatizing rather
than defining them) is that it makes possible to enrich the calculus
with extra instructions and evaluation rules, simply by putting
additional axioms about~$\C$, $\B$ and~$\eval_1$.
On the other hand, the definitions of classical
realizability~\cite{Kri09} as well as its main properties do not
depend on the particular choice of~$\B$, $\C$ and~$\eval_1$, although
the fine structure of the corresponding realizability models is of
course affected by the presence of additional instructions and
evaluation rules.

For the needs of the discussion in Section~\ref{s:subst}, we shall
sometimes consider the following extra instructions in the set~$\C$:
\begin{itemize}
\item The instruction $\Quote$, which comes with the evaluation
  rule
  $$\Quote\star t\cdot\pi~\eval_1~t\star\barr{n}_{\pi}\cdot\pi\,,
  \leqno(\textsc{Quote})$$
  where $\pi\mapsto n_{\pi}$ is a recursive injection from~$\Pi$
  to~$\N$.
  Intuitively, the instruction $\Quote$ computes the `code' $n_{\pi}$
  of the stack~$\pi$, and passes it (using the encoding
  $n\mapsto\barr{n}$ described in Section~\ref{ss:TermsStacks}) to the
  term~$t$.
  This instruction was originally introduced  to realize the axiom
  of dependent choices~\cite{Kri03}.
  \medbreak
\item The instruction $\Eq$, which comes with the evaluation
  rule
  $$\Eq\star t_1\cdot t_2\cdot u\cdot v\cdot\pi
  ~\eval_1~\begin{cases}
    u\star\pi & \text{if}~t_1\equiv t_2 \\
    v\star\pi & \text{if}~t_1\not\equiv t_2 \\
  \end{cases}\leqno(\textsc{Eq})$$
  Intuitively, the instruction $\Eq$ tests the syntactic equality of
  its first two arguments~$t_1$ and~$t_2$ (up to $\alpha$-conversion),
  giving the control to the next argument~$u$ if the test succeeds,
  and to the second next argument~$v$ otherwise.
  In presence of the $\Quote$ instruction, it is possible to implement
  a closed~$\lambda_c$-term $\Eq'$ that has the very same
  computational behaviour as~$\Eq$, by letting
  $$\Eq'~\equiv~
  \Lam{x_1x_2}{\Quote\,(\Lam{n_1y_1}{\Quote\,(\Lam{n_2y_2}
      {\mathsf{eq\_nat}\,n_1\,n_2})\,x_2})\,x_1}\,,$$
  where $\mathsf{eq\_nat}$ is any closed $\lambda$-term that tests the
  equality between two numerals (using the encoding $n\mapsto\barr{n}$).
  \medbreak
\item The instruction $\Fork$ (`fork'), which comes with the two
  evaluation rules
  $$\Fork\star t_0\cdot t_1\cdot\pi\eval_1 t_0\star\pi
  \qquad\text{and}\qquad
  \Fork\star t_0\cdot t_1\cdot\pi\eval_1 t_1\star\pi\,.
  \leqno(\textsc{Fork})$$
  Intuitively, the instruction~$\Fork$ behaves as a non deterministic
  choice operator, that indifferently selects its first or its second
  argument.
  The main interest of this instruction is that it makes evaluation
  non deterministic, in the following sense:
\end{itemize}

\begin{definition}[Deterministic evaluation]\label{d:Determinism}
 We say that the relation of evaluation $\eval_1$ is
  \emph{deterministic} when the two conditions $p\eval_1p'$ and
  $p\eval_1p''$ imply $p'\equiv p''$ (syntactic identity) for all
  processes $p$, $p'$ and~$p''$.
  Otherwise, $\eval_1$ is said to be \emph{non deterministic}.
\end{definition}

The smallest relation of evaluation, that is defined as the
union of the four rules (\textsc{Push}), (\textsc{Grab}),
(\textsc{Save}) and (\textsc{Restore}), is clearly deterministic.
The property of determinism still holds if we enrich the calculus with
an instruction~$\Eq\,({\not\equiv}\,\cc)$ together with the
aforementioned evaluation rules or with the
instruction~$\Quote\,({\not\equiv}\,\cc)$.

On the other hand, the presence of an instruction~$\Fork$ with the
corresponding evaluation rules definitely makes the relation of
evaluation non deterministic.

\subsection{The thread of a process and its anatomy}
Given a process~$p$, we call the \emph{thread} of~$p$ and
write~$\Th(p)$ the set of all processes~$p'$ such that $p\eval p'$:
$$\Th(p)~=~\{p'\in\Lambda\star\Pi~:~p\eval p'\}\,.$$
This set has the structure of a finite or infinite (di)graph whose
edges are given by the relation~$\eval_1$ of one step evaluation.
In the case where the relation of evaluation is deterministic, the
graph $\Th(p)$ can be either:
\begin{itemize}
\item\emph{Finite and cyclic from a certain point}, because the
  evaluation of~$p$ loops at some point.
  A typical example is the process
  $\textbf{I}\star\delta\delta\cdot\alpha$ (where
  $\textbf{I}\equiv\Lam{x}{x}$ and $\delta\equiv\Lam{x}{xx}$),
  that enters into a 2-cycle after one evaluation step:
  $$\textbf{I}\star\delta\delta\cdot\alpha
  ~\eval_1~\delta\delta\star\alpha
  ~\eval_1~\delta\star\delta\cdot\alpha~
  ~\eval_1~\delta\delta\star\alpha
  ~\eval_1~\cdots$$
\item\emph{Finite and linear}, because the evaluation of~$p$ reaches a
  state where no more rule applies.
  For example:
  $$\textbf{I}\textbf{I}\star\alpha
  ~\eval_1~\textbf{I}\star\textbf{I}\cdot\alpha
  ~\eval_1~\textbf{I}\star\alpha\,.$$
\item\emph{Infinite and linear}, because~$p$ has an infinite execution
  that never reaches twice the same state.
  A typical example is given by the process
  $\delta'\delta'\star\alpha$, where
  $\delta'\equiv\Lam{x}{x\,x\,\textbf{I}}$:
  $$\delta'\delta'\star\alpha
  ~\eval_3~\delta'\delta'\star\textbf{I}\cdot\alpha
  ~\eval_3~\delta'\delta'\star\textbf{I}\cdot\textbf{I}\cdot\alpha
  ~\eval_3~\delta'\delta'\star\textbf{I}\cdot\textbf{I}
  \cdot\textbf{I}\cdot\alpha
  ~\eval_3~\cdots$$
\end{itemize}



\subsection{Interaction constants}
The two examples of extra instructions $\Quote$ and $\Eq$ we gave in Section \ref{ss:ExtraInstr}
have a strong impact on the potential behaviour of processes.
Indeed, they are able to distinguish syntactically different terms that are computationally equivalent,
such as the terms $\I$ and $\I\I$.
To better understand the consequence of the presence of such extra instructions in the $\lambda_c$-calculus,
we need to introduce the important notion of interaction constant.   
This definition relies on the notions of substitution over terms and stacks, that are defined in Fig.~\ref{fig:bnf}.
Unlike the traditional form of substitution $t\{x:=u\}$ (which is only
defined for terms), the substitutions $t\{c:=u\}$ and $\pi\{c:=u\}$ also
propagate through the continuation constants $\k_{\pi}$.

\begin{definition}
 A constant $\kappa\in\C$ is said to be
 \begin{itemize}
  \item \emph{inert} if for all $\pi\in\Pi$, there is no process $p$ such that $\kappa\star\pi\eval_1 p$;
  \item \emph{substitutive} if for all $u\in\Lambda$ and for all processes $p,p'\in\Lambda\star\Pi$, 
  $p\eval_1 p'$ implies $p\{\kappa:=u\}\eval_1 p'\{\kappa:=u\}$;
  \item \emph{non generative} if for all processes $p,p'\in\Lambda\star\Pi$, 
  $p\eval_1 p'$, the constant $\kappa$ cannot occur in $p'$ unless it already occurs in $p$.
 \end{itemize}
 A constant $\kappa\in\C$ that is inert, substitutive and non
 generative is then called an \emph{interaction constant}.  
 Similarly, we say that a stack constant $\alpha\in\B$ is:
\begin{itemize}
    \item \emph{substitutive} if for all $\pi\in\Pi$ and for all
      processes $p,p'\in\Lambda\star\Pi$,  
  $p\eval_1 p'$ implies $p\{\alpha:=\pi\}\eval_1 p'\{\alpha:=\pi\}$;
  \item \emph{non generative} if for all processes $p,p'\in\Lambda\star\Pi$, 
  $p\eval_1 p'$, the constant $\alpha$ cannot occur in $p'$ unless it already occurs in $p$.
 \end{itemize}
\end{definition}

The main observation is that substitutive constants are incompatible
with both instruction $\Quote$ and $\Eq$ (see \cite{Peirce} for a proof): 
\begin{prop}\label{prop:ExtraInstructions}
 If the calculus of realizers contains one of both instructions $\Quote$ or $\Eq$, 
 then none of the constants $\kappa\in\C$ is substitutive.
\end{prop}

The very same argument can be applied to prove the incompatibility of substitutive stack constants with the instruction $\Quote$.
On the other hand, it is clear that if the relation of evaluation $\eval_1$ is only defined from the rules
(\textsc{Grab}),(\textsc{Push}),(\textsc{Save}) and (\textsc{Restore}) -and possibly: the rule (\textsc{Fork})- then all the remaining
constants $\kappa$ in $\C$ (i.e. $\kappa\nequiv \cc,\Fork$) are interaction constants (and thus substitutive),
whereas all the stack constants in $\B$ are substitutive and non generative.
Substitutive (term and stack) constants are useful to analyze the computational behaviour of realizers in 
a uniform way. For instance, if we know that a closed term $t\in\Lambda$ is such that
$$t\star \kappa_1\cdots\kappa_n\cdot\alpha\eval p$$
where $\kappa_1,\ldots,\kappa_n$ are substitutive constants that do not occur in $t$, and where $\alpha$ is a substitutive
stack constant that does not occur in $t$ too,
then we more generally know that
$$t\star u_1\cdots u_n\cdot\pi \eval p\{\kappa_1:=u_1,\ldots,\kappa_n:=u_n,\alpha:=\pi\}$$
for all terms $u_1,\ldots,u_n\in\Lambda$ and for all stacks $\pi\in\Pi$. 
Intuitively, substitutive constants play in the $\lambda_c$-calculus the same role as free variables in the pure $\lambda$-calculus.


\section{Classical second-order arithmetic}
\label{s:PA2}

In Section~\ref{s:Lamc} we delt with the \emph{computing facet}
of the theory of classical realizability.
In this section, we will now present its \emph{logical facet} by
introducing the language of classical second-order logic with the
corresponding type system.
In section~\ref{ss:PA2}, we will focus to the particular case of
\emph{second-order arithmetic} and present its axioms.

\subsection{The language of second-order logic}
\label{s:LangSOL}

The language of second-order logic distinguishes two kinds of
expressions:
\emph{first-order expressions} %
 representing individuals, and \emph{formul{\ae}}, representing
propositions about individuals and sets of individuals (represented
using second-order variables as we shall see below).

\subsubsection{First-order expressions}\label{sss:FOExpr}
First-order expressions are formally defined (see Fig.~\ref{fig:bnf}) from the following sets of
symbols:
\begin{itemize}
\item A \emph{first-order signature}~$\Sigma$ defining \emph{function
  symbols} with their arities, and considering \emph{constant symbols}
  as function symbols of arity~$0$.
  We assume that the signature~$\Sigma$ contains a
constant symbol~$0$ (`zero'), a unary function symbol~$s$ (`successor')
as well as a function symbol~$f$ for every primitive recursive
function (including symbols $+$, $\times$, etc.), each of them
being given its standard interpretation in~$\N$ (see Section \ref{ss:PA2}).

\item A denumerable set~$\V_1$ of \emph{first-order variables}.
  For convenience, we shall still use the lowercase letters $x$, $y$,
  $z$, etc.\ to denote first-order variables, but these variables
  should not be confused with the $\lambda$-variables introduced in
  Section~\ref{s:Lamc}.
\end{itemize}


The set~$\FV(e)$ of all (free) variables of a first-order
expression~$e$ is defined as expected, as well as the corresponding
operation of substitution, that we still write $e\{x:=e'\}$.

\subsubsection{Formul\ae}\label{sss:Formulas}
Formul{\ae} of second-order logic are defined (see Fig.~\ref{fig:bnf}) from an additional set
of symbols~$\V_2$ of \emph{second-order variables} (or \emph{predicate
  variables}), using the uppercase letters~$X$, $Y$, $Z$, etc.\ to
represent such variables:
$$A,B::=X(e_1,\ldots,e_k)\mid A\limp B \mid \forall x A\mid \forall X A\eqno(X\in\V_2)$$
We assume that each second-order variable~$X$ comes with an arity
$k\ge 0$ (that we shall often leave implicit since it can be easily
inferred from the context), and that for each arity $k\ge 0$, the
subset of~$\V_2$ formed by all second-order variables of arity~$k$ is
denumerable.

Intuitively, second-order variables of arity~$0$ represent (unknown)
propositions, unary predicate variables represent predicates over
individuals (or \emph{sets} of individuals) whereas binary predicate
variables represent binary relations (or sets of pairs), etc.


The set of free variables of a formula~$A$ is written~$\FV(A)$.
(This set may contain both first-order and second-order variables.)
As usual, formul{\ae} are identified up to $\alpha$-conversion,
neglecting differences in bound variable names.
Given a formula~$A$, a first-order variable~$x$ and a closed
first-order expression~$e$, we denote by $A\{x:=e\}$ the formula
obtained by replacing every free occurrence of~$x$ by the first-order
expression~$e$ in the formula~$A$, possibly renaming some bound
variables of~$A$ to avoid name clashes.

Lastly, although the formul{\ae} of the language of second-order logic are
constructed from atomic formul{\ae} only using implication and first- and
second-order universal quantifications, we can define other logical
constructions (negation, conjunction disjunction, first- and
second-order existential quantification as well as Leibniz equality)
using the so called second-order encodings (cf Fig.~\ref{fig:bnf}).

\subsubsection{Predicates and second-order substitution}
We call a \emph{predicate of arity~$k$} any expression of the form
$P\equiv\Lam{x_1\cdots x_k}{C}$ where $x_1,\ldots,x_k$ are $k$
pairwise distinct first-order variables and where~$C$ is an arbitrary
formula. (Here, we (ab)use the $\lambda$-notation to indicate which variables
$x_1,\ldots,x_k$ are abstracted in the formula~$C$).

The set of free variables of a $k$-ary predicate
$P\equiv\Lam{x_1\cdots x_k}{C}$ is defined by
$\FV(P)\equiv\FV(C)\setminus\{x_1;\ldots;x_k\}$, and the application
of the predicate $P\equiv\Lam{x_1\cdots x_k}{C}$ to a $k$-tuple
of first-order expressions $e_1,\ldots,e_k$ is defined by letting
$$P(e_1,\ldots,e_k)~\equiv~
(\Lam{x_1\cdots x_k}{C})(e_1,\ldots,e_k)~\equiv~
C\{x_1:=e_1;\ldots;x_k:=e_k\}$$
(by analogy with $\beta$-reduction).
Given a formula~$A$, a $k$-ary predicate variable~$X$ and an actual
$k$-ary predicate~$P$, we finally define the operation of
\emph{second-order substitution} $A\{X:=P\}$ as follows:
$$\begin{array}{rcl}
  X(e_1,\ldots,e_k)\{X:=P\}
  &\equiv& P(e_1,\ldots,e_k) \\
  Y(e_1,\ldots,e_m)\{X:=P\} &\equiv& Y(e_1,\ldots,e_m) \\
  (A\limp B)\{X:=P\} &\equiv&
  A\{X:=P\}\limp
  B\{X:=P\} \\
  (\forall x\,A)\{X:=P\} &\equiv&
  \forall x\,A\{X:=P\} \\
  (\forall X\,A)\{X:=P\} &\equiv&
  \forall X\,A \\
  (\forall Y\,A)\{X:=P\} &\equiv&
  \forall Y\,A\{X:=P\} \\
\end{array}\eqno\begin{array}{r@{}}
\\(Y\not\equiv X)\\\\(x\notin\FV(P))\\\\
(Y\not\equiv X,~Y\notin\FV(P))\\
\end{array}$$


\subsection{A type system for classical second-order logic}

Through the formul{\ae}-as-types correspondence~\cite{How69,Gir89}, we
can see any formula~$A$ of second-order logic as a type, namely, as
the type of its proofs.
We shall thus present the deduction system of classical second-order
logic as a type system based on a typing judgement of the form
$\TYP{\Gamma}{t}{A}$, where
\begin{itemize}
\item $\Gamma$ is a typing context of the form
  $\Gamma\equiv x_1:B_1,\ldots,x_n:B_n$, where $x_1,\ldots,x_n$ are
  pairwise distinct $\lambda$-variables and where $B_1,\ldots,B_n$ are
  arbitrary propositions;
\item $t$ is a proof-like term, i.e.\ a $\lambda_c$-term
  containing no continuation constant~$\k_{\pi}$;
\item $A$ is a formula of second-order logic.
\end{itemize}

The type system of classical second-order logic is then defined from
the typing rules of Fig.~\ref{fig:bnf}.
These typing rules are the usual typing rules of AF2~\cite{Kri93},
plus a specific typing rule for the instruction~$\cc$ which permits to
recover the full strength of classical logic.

Using the encodings of second-order logic, we can derive from
the typing rules of Fig.~\ref{fig:bnf} the usual introduction and
elimination rules of absurdity, conjunction, disjunction, (first- and
second-order) existential quantification and Leibniz
equality~\cite{Kri93}.
The typing rule for call$/$cc (law of Peirce) allows us to
construct proof-terms for classical reasoning principles such as the
excluded middle, \emph{reductio ad absurdum}, de Morgan laws, etc.

\subsection{Classical second-order arithmetic (PA2)}
\label{ss:PA2}

From now on, we consider the particular case of \emph{second-order
  arithmetic} (PA2), where first-order expressions are intended to
represent natural numbers.
For that, we assume that every $k$-ary function symbol $f\in\Sigma$
comes with an interpretation in the standard model of arithmetic as a
function $\Int{f}:\N^k\to\N$, so that we can give a denotation
$\Int{e}\in\N$ to every closed first-order expression~$e$.
Moreover, we assume that each function symbol associated to a primitive
recursive definition (cf Section~\ref{sss:FOExpr}) is given its standard interpretation
in~$\N$.
In this way, every numeral $n\in\N$ is represented in the world of
first-order expressions as the closed expression $s^n(0)$ that we
still write~$n$, since $\Int{s^n(0)}=n$.

\subsubsection{Induction}
\label{sss:Induction}

Following Dedekind's construction of natural numbers, we consider the
predicate $\Nat(x)$~\cite{Gir89,Kri93} defined by
$$\Nat(x)~\equiv~\forall Z\,(Z(0)\limp
\forall y\,(Z(y)\limp Z(s(y)))\limp Z(x))\,,$$
that defines the smallest class of individuals containing zero and
closed under the successor function.
One of the main properties of the logical system presented above is
that the axiom of induction, that we can write $\forall x~\Nat(x)$, is
not derivable from the rules of Fig.~\ref{fig:bnf}.
As Krivine proved~\cite[Theorem~12]{Kri09}, this axiom is not even
(universally) realizable in general.
To recover the strength of arithmetic reasoning, we need to relativize
all first-order quantifications to the class $\Nat(x)$ of Dedekind
numerals using the shorthands for \emph{numeric quantifications}%
$$\begin{array}{rcl}
  \forallNat x\,A(x) &\equiv& \forall x\,(\Nat(x)\limp A(x)) \\
  \existsNat x\,A(x) &\equiv&
  \forall Z\,(\forall x(\Nat(x)\limp A(x)\limp Z)\limp Z) \\
\end{array}$$
so that the \emph{relativized induction axiom} becomes
provable in second-order logic~\cite{Kri93}:
$$\forall Z\,(Z(0)\limp\forallNat x\,(Z(x)\limp Z(s(x)))\limp
\forallNat x Z(x))\,.$$

\subsubsection{The axioms of PA2}
\label{sss:AxiomsPA2}

Formally, a formula~$A$ is a \emph{theorem} of second-order arithmetic
(PA2) if it can be derived (using the rules of Fig.~\ref{fig:bnf})
from the two axioms
\begin{itemize}
\item$\forall x\,\forall y\,(s(x)=s(y)\limp x=y)$\hfill
  (Peano 3rd axiom)
\item$\forall x\,\lnot(s(x)=0)$\hfill (Peano 4th axiom)
\end{itemize}
expressing that the successor function is injective and not
surjective, and from the definitional equalities attached to the
(primitive recursive) function symbols of the signature:
\begin{itemize}
\item$\forall x\,(x+0=x)$,\quad
  $\forall x\,\forall y\,(x+s(y)=s(x+y))$
\item$\forall x\,(x\times 0=0)$,\quad
  $\forall x\,\forall y\,(x\times s(y)=(x\times y)+x)$
\item etc.
\end{itemize}
Unlike the non relativized induction axiom---that requires a special
treatment in~PA2---we shall see in Section~\ref{ss:RealizAxioms} that
all these axioms are realized by simple proof-like terms.

\section{Classical realizability semantics}
\label{s:Realiz}

\subsection{Generalities}
\label{ss:Generalities}

Given a particular instance of the $\lambda_c$-calculus (defined from
particular sets~$\B$, $\C$ and from a particular relation of
evaluation~$\eval_1$ as described in Section~\ref{s:Lamc}), we shall
now build a classical realizability model in which every closed
formula~$A$ of the language of~PA2 will be interpreted as a set of
closed terms $|A|\subseteq\Lambda$, called the \emph{truth value}
of~$A$, and whose elements will be called the \emph{realizers} of~$A$.

\subsubsection{Poles, truth values and falsity values}
\label{sss:Poles}
Formally, the construction of the realizability model is parameterized
by a \emph{pole}%
$\Bot$ in the sense of the following definition:
\begin{definition}[Poles]
  --- A \emph{pole} is any set of processes
  $\Bot\subseteq\Lambda\star\Pi$ which is closed under
  anti-evaluation, in the sense that both conditions $p\eval p'$ and
  $p'\in\Bot$ together imply that $p\in\Bot$ for all processes
  $p,p'\in\Lambda\star\Pi$.
\end{definition}
We will mainly use one method to define a pole~$\Bot$.
From an arbitrary
set of processes $P$,
we can define $pole$ as the complement 
  set of the union of all threads starting from an element of~$P$,
  that is:
  $$\Bot~\equiv~\biggl(\bigcup_{p\in P}\Th(p)\biggr)^c
  ~\equiv~\bigcap_{p\in P}\bigl(\Th(p)\bigr)^c\,.$$
  It is indeed quite easy to check that $\pole$ is closed by anti-reduction,
  and it is also the largest pole that does not
  intersect~$P$. We shall say that such a definition is
  \emph{thread-oriented}.

Let us now consider a fixed pole~$\Bot$.
We call a \emph{falsity value} any set of stacks $S\subseteq\Pi$.
Every falsity value $S\subseteq\Pi$ induces a \emph{truth value}
$S^{\Bot}\subseteq\Lambda$ that is defined by
$$S^{\Bot}~=~\{t\in\Lambda~:~
\forall\pi\in S~(t\star\pi)\in\Bot\}\,.$$
Intuitively, every falsity value $S\subseteq\Pi$ represents a
particular set of \emph{tests}, while the corresponding truth value
$S^{\Bot}$ represent the set of all \emph{programs} that passes all
tests in~$S$ (w.r.t.\ the pole~$\Bot$, that can be seen as the
\emph{challenge}).
From the definition of $S^{\Bot}$, it is clear that the larger the
falsity value~$S$, the smaller the corresponding truth
value~$S^{\Bot}$, and vice-versa.


\subsubsection{Formul{\ae} with parameters}
In order to interpret second-order variables that occur in a given
formula~$A$, it is convenient to enrich the language of PA2 with a new
predicate symbol $\dot{F}$ of arity~$k$ for every \emph{falsity value
  function~$F$} of arity~$k$, that is, for every function
$F:\N^{k}\to\Pow(\Pi)$ that associates a falsity value
$F(n_1,\ldots,n_k)\subseteq\Pi$ to every $k$-tuple
$(n_1,\ldots,n_k)\in\N^k$.
A formula of the language enriched with the predicate
symbols~$\dot{F}$ is then called a \emph{formula with parameters}.
Formally, this correspond to the formul\ae\ defined by:  
$$A,B::=X(e_1,\ldots,e_k)\mid A\limp B \mid \forall x A\mid \forall X A\mid \dot{F}(e_1,\ldots,e_k)\eqno X\in\V_2,F\in\Pow(\Pi)^{\N^k}$$

The notions of a \emph{predicate with parameters} and of a
\emph{typing context with parameters} are defined similarly.
The notations $\FV(A)$, $\FV(P)$, $\FV(\Gamma)$, $\dom(\Gamma)$,
$A\{x:=e\}$, $A\{X:=P\}$, etc.\ are extended to all formul{\ae}~$A$ with
parameters, to all predicates~$P$ with parameters and to all typing
contexts~$\Gamma$ with parameters in the obvious way.

\subsection{Definition of the interpretation function}
\label{ss:Interp}

The interpretation of the closed formul{\ae} with parameters is defined
as follows:
\begin{definition}[Interpretation of closed formul{\ae} with parameters]
  --- The falsity value $\|A\|\subseteq\Pi$ of a closed formula~$A$
  with parameters is defined by induction on the number of
  connectives$/$quantifiers in~$A$ from the equations
  $$\begin{array}{rcl}
    \|\dot{F}(e_1,\ldots,e_k)\| &=&
    F(\Int{e_1},\ldots,\Int{e_k}) \\
    \noalign{\medskip}
    \|A\limp B\| &=&
    |A|\cdot\|B\|\quad={\quad}
    \bigl\{t\cdot\pi~:~ t\in|A|,~\pi\in\|B\|\bigr\} \\
    \noalign{\medskip}
    \|\forall x\,A\| &=& \ds
    \bigcup_{n\in\N}\|A\{x:=n\}\| \\
    \noalign{\medskip}
    \|\forall X\,A\| &=& \ds
    \bigcup_{\!\!\!F:\N^k\to\Pow(\Pi)\!\!\!}\|A\{X:=\dot{F}\}\|
    \qquad(\text{if}~X~\text{has arity}~k)\\
  \end{array}$$
  whereas its truth value $|A|\subseteq\Lambda$ is defined by
  $|A|=\|A\|^{\Bot}$. 
  Finally, defining $\top\equiv\dot{\emptyset}$ (recall that we have $\bot\equiv\forall X.X$), one can check that we have :
  $$\|\top\|=\emptyset \qquad\quad |\top|=\Lambda \qquad\qquad \|\bot\| = \Pi  $$

\end{definition}

Since the falsity value~$\|A\|$ (resp.\ the truth value~$|A|$) of~$A$
actually depends on the pole~$\Bot$, we shall write it sometimes
$\|A\|_{\Bot}$ (resp.\ $|A|_{\Bot}$) to recall the dependency.
Given a closed formula~$A$ with parameters and a closed
term~$t\in\Lambda$, we say that:
\begin{itemize}
\item\emph{$t$ realizes~$A$} and write $t\real A$ when
  $t\in|A|_{\Bot}$.\\
  (This notion is relative to a particular pole~$\Bot$.)
\item\emph{$t$ universally realizes~$A$} and write
  $t\ureal A$ when $t\in|A|_{\Bot}$ for all poles~$\Bot$.
\end{itemize}

From these definitions, 
we have

\begin{lem}[Law of Peirce]\label{l:Peirce}
  --- Let~$A$ and~$B$ be two closed formul{\ae} with parameters:
  \begin{enumerate}
  \item If $\pi\in\|A\|$, then $\k_{\pi}\real A\limp B$.
  \item $\cc\ureal((A\limp B)\limp A)\limp A$.
  \end{enumerate}
\end{lem}


\subsection{Valuations and substitutions}
In order to express the soundness invariants relating the type system
of Section~\ref{s:PA2} with the classical realizability semantics
defined above, we need to introduce some more terminology.

\begin{definition}[Valuations]
  --- A \emph{valuation} is a function~$\rho$ that associates a
  natural number $\rho(x)\in\N$ to every first-order variable~$x$ and
  a falsity value function $\rho(X):\N^k\to\Pow(\Pi)$ to every
  second-order variable~$X$ of arity~$k$.
  \begin{itemize}
  \item Given a valuation~$\rho$, a first-order variable~$x$ and a
    natural number $n\in\N$, we denote by $\rho,x\gets n$ the
    valuation defined by:
    $$(\rho,x\gets n)~=~
    \rho_{|\dom(\rho)\setminus\{x\}}\cup\{x\gets n\}\,.$$
  \item Given a valuation~$\rho$, a second-order variable~$X$ of
    arity~$k$ and a falsity value function $F:\N^k\to\Pow(\Pi)$,
    we denote by $\rho,X\gets F$ the valuation defined by:
    $$(\rho,X\gets F)~=~
    \rho_{|\dom(\rho)\setminus\{X\}}\cup\{X\gets F\}\,.$$
  \end{itemize}
\end{definition}

To every pair $(A,\rho)$ formed by a (possibly open) formula~$A$
of~PA2 and a valuation~$\rho$, we associate a \emph{closed} formula
with parameters $A[\rho]$ that is defined by
$$A[\rho]~\equiv~A\{x_1:=\rho(x_1);\ldots;x_n:=\rho(x_n);
X_1:=\dot{\rho}(X_1);\ldots;X_m:=\dot{\rho}(X_m)\}$$
where $x_1,\ldots,x_n,X_1,\ldots,X_m$ are the free variables of~$A$,
and writing $\dot{\rho}(X_i)$ the predicate symbol associated to the
falsity value function $\rho(X_i)$.
This operation naturally extends to typing contexts by letting
$(x_1:A_1,\ldots,x_n:A_n)[\rho]\equiv
x_1:A_1[\rho],\ldots,x_n:A_n[\rho]$.

\begin{definition}[Substitutions]
  --- A \emph{substitution} is a finite function~$\sigma$ from
  $\lambda$-variables to closed $\lambda_c$-terms.
  Given a substitution~$\sigma$, a $\lambda$-variable~$x$ and a closed
  $\lambda_c$-term~$u$, we denote by $\sigma,x:=u$ the substitution
  defined by
  $(\sigma,x:=u)\equiv\sigma_{|\dom(\sigma)\setminus\{x\}}\cup\{x:=u\}$.
\end{definition}

Given an open $\lambda_c$-term~$t$ and a substitution~$\sigma$, we
denote by $t[\sigma]$ the term defined by
$$t[\sigma]~\equiv~t\{x_1:=\sigma(x_1);\ldots;x_n:=\sigma(x_n)\}$$
where $\dom(\sigma)=\{x_1,\ldots,x_n\}$.
Notice that $t[\sigma]$ is closed as soon as
$\FV(t)\subseteq\dom(\sigma)$.
We say that a substitution~$\sigma$ \emph{realizes} a closed
context~$\Gamma$ with parameters and write $\sigma\real\Gamma$ if:
\begin{itemize}
\item$\dom(\sigma)=\dom(\Gamma)$;
\item$\sigma(x)\real A$ for every declaration $(x:A)\in\Gamma$.
\end{itemize}

\subsection{Adequacy}
\label{ss:Adequacy}

Given a fixed pole~$\Bot$, we say that:
\begin{itemize}
\item A typing judgement $\TYP{\Gamma}{t}{A}$ is \emph{adequate}
  (w.r.t.\ the pole~$\Bot$) if for all valuations $\rho$ and for all
  substitutions $\sigma\real\Gamma[\rho]$ we have
  $t[\sigma]\real A[\rho]$.
\item More generally, we say that an inference rule
  $$\infer{J_0}{J_1 &\cdots& J_n}$$
  is adequate (w.r.t.\ the pole~$\Bot$) if the adequacy of all typing
  judgements $J_1,\ldots,J_n$ implies the adequacy of the typing
  judgement $J_0$.
\end{itemize}
From the latter definition, it is clear that a typing judgement that is
derivable from a set of adequate inference rules is adequate too.

\begin{prop}[Adequacy~\cite{Kri09}]\label{p:Adequacy}
  The typing rules of Fig.~\ref{fig:bnf} are adequate
  w.r.t.\ any pole~$\Bot$, as well as all the judgements
  $\TYP{\Gamma}{t}{A}$ that are derivable from these rules.
\end{prop}


Since the typing rules of Fig.~\ref{fig:bnf} involve no continuation
constant, every realizer that comes from a proof of second order logic
by Prop.~\ref{p:Adequacy} is thus a proof-like term.

\subsection{Realizing the axioms of PA2}
\label{ss:RealizAxioms}

Let us recall that in PA2, Leibniz equality $e_1=e_2$ is defined by
$e_1=e_2\equiv\forall Z\,(Z(e_1)\limp Z(e_2))$.

\begin{prop}[Realizing Peano axioms~\cite{Kri09}]
  \label{p:PeanoAxioms}~:
  \begin{enumerate}
  \item $\Lam{z}{z}~\ureal~
    \forall x\,\forall y\,(s(x)=s(y)\limp x=y)$
  \item $\Lam{z}{zu}~\ureal~
    \forall x\,(s(x)=0\limp\bot)$\hfill
    (where~$u$ is any term such that $\FV(u)\subseteq\{z\}$).
  \item $\Lam{z}{z}~\ureal~
    \forall x_1\cdots\forall x_k\,
    (e_1(x_1,\ldots,x_n)=e_2(x_1,\ldots,x_k))$\\
    for all arithmetic expressions
    $e_1(x_1,\ldots,x_n)$ and $e_2(x_1,\ldots,x_k)$ such that\\
    $\N\models \forall x_1\cdots\forall x_k\,
    (e_1(x_1,\ldots,x_n)=e_2(x_1,\ldots,x_k))$.
  \end{enumerate}
\end{prop}


From this we deduce the main theorem:
\begin{thm}[Realizing the theorems of~PA2]\label{th:RealizPA2}
  --- If~$A$ is a theorem of PA2 (in the sense defined in
  Section~\ref{sss:AxiomsPA2}), then there is a closed proof-like
  term~$t$ such that $t\ureal A$.
\end{thm}

\begin{proof}
  Immediately follows from Prop.~\ref{p:Adequacy}
  and~\ref{p:PeanoAxioms}.
\end{proof}

\subsection{The full standard model of PA2 as a degenerate case}
\label{ss:Degenerate}

It is easy to see that when the pole~$\Bot$ is empty, the classical
realizability model defined above collapses to the \emph{full standard
  model of PA2}, that is: to the model (in the sense of Tarski)
where individuals are interpreted by the elements of~$\N$ and where
second-order variables of arity~$k$ are interpreted by all the subsets
of~$\N^k$.
For that, we first notice that when
$\Bot=\varnothing$, the truth value~$S^{\Bot}$ associated to an
arbitrary falsity value $S\subseteq\Pi$ can only take two different
values: $S^{\Bot}=\Lambda_c$ when $S=\varnothing$, and
$S^{\Bot}=\varnothing$ when $S\neq\varnothing$.
Moreover, we easily check that the realizability interpretation of
implication and universal quantification mimics the standard truth
value interpretation of the corresponding logical construction
in the case where~$\Bot=\varnothing$.
Writing~$\M$ for the full standard model of PA2, we thus easily show
that:
\begin{prop}\label{p:Degenerated}
  --- If $\Bot=\varnothing$, then for every closed formula~$A$ of PA2
  we have
  $$|A|=\begin{cases}
    \Lambda &\text{if}~\M\models A \\
    \varnothing &\text{if}~\M\not\models A \\
  \end{cases}$$
\end{prop}

\begin{proof}
  We more generally show that for all formul{\ae}~$A$ and for all
  valuations~$\rho$ closing~$A$ (in the sense defined in
  section~\ref{ss:Interp}) we have
  $$|A[\rho]|=\begin{cases}
    \Lambda &\text{if}~\M\models A[\tilde{\rho}] \\
    \varnothing &\text{if}~\M\not\models A[\tilde{\rho}] \\
  \end{cases}$$
  where $\tilde{\rho}$ is the valuation in~$\M$ (in the usual sense)
  defined by
  \begin{itemize}
  \item $\tilde{\rho}(x)=\rho(x)$ for all first-order variables~$x$;
  \item $\tilde{\rho}(X)=\{(n_1,\ldots,n_k)\in\N^k:
    \rho(X)(n_1,\ldots,n_k)=\varnothing\}$ for all second-order
    variables~$X$ of arity $k$.
  \end{itemize}
  (This characterization is proved by a straightforward induction
  on~$A$.)
\end{proof}

An interesting consequence of the above lemma is the following:
\begin{cor}
  --- If a closed formula~$A$ has a universal realizer $t\ureal A$,
  then~$A$ is true in the full standard model~$\M$ of PA2.
\end{cor}

\begin{proof}
  If $t\ureal A$, then $t\in|A|_{\varnothing}$.
  Therefore $|A|_{\varnothing}=\Lambda$ and~$\M\models A$.
\end{proof}

However, the converse implication is false in general, since the
formula $\forall x\,\Nat(x)$ (cf Section~\ref{sss:Induction}) that
expresses the induction principle over individuals is obviously true
in~$\M$, but it has no universal realizer when evaluation is
deterministic~\cite[Theorem~12]{Kri09}.

\subsection{Relativization to canonical integers}
We previously explained in Section \ref{sss:Induction} that we needed to relativize first-order quantifications to the class $\Nat(x)$.
If we have as expected $\bar n\ureal \Nat(n)$ for any $n\in\N$, there are realizers of $\Nat(n)$ different from is $\bar n$. 
Intuitively, a term $t\ureal\Nat(n)$ represents the integer $n$, but $n$ might be present only as a computation, and not directly as a computed value.

The usual technique to retrieve $\bar n$ from such a term consist in the use of a storage operator $T$, which would make our definition of game harder.
Rather than that, we define a new asymmetrical implication where the left member must be an integer value, and the interpretation of this new implication.
$$\begin{array}{r@{~}c@{~}l}
 A,B~ & ::= &~ \ldots \mid \{e\}\limp A \\[0.2em]
 \|\{e\}\limp A\| & = & \{\bar n\cdot \pi : \Int{e}=n \wedge \pi\in\|A\|\}\\
\end{array}\leqno\begin{array}{l}
\textrm{\bf Formul\ae} \\ [0.2em]
\textrm{\bf Falsity value} \\ 
\end{array}$$
We finally define the corresponding shorthands for relativized quantifications:
$$\begin{array}{rcl}
   \forallN x\,A(x) &\equiv& \forall x\,(\{x\}\limp A(x)) \\
 \existsN x\,A(x) & \equiv &   \forall Z\,(\forall x(\{x\}\limp A(x)\limp Z)\limp Z) \\
\end{array}$$
It is easy to check that this relativization of first-order quantification is equivalent (in terms of realizability) to the one defined in Section \ref{sss:Induction} and 
that the relativized principle of induction holds.
\begin{prop}
 Let $T$ be a storage operator. The following holds for any formula $A(x)$:
 \begin{enumerate}
  \item $\lambda x.x\ureal \forallN x \Nat(x)$
  \item $\lambda x.x \ureal \forallNat x.A(x) \limp \forallN x.A(x)$
  \item $\lambda x.Tx \ureal \forallN x.A(x) \limp \forallNat x.A(x)$
  \end{enumerate}
\end{prop}

For further details about the relativization and storage operator, please refer to Section 2.9 and 2.10.1 of Rieg's Ph.D. thesis~\cite{RiegPhD}.

\subsection{Leibniz equality}

Before going further, we would like to draw the reader's attention to the treatment 
that is given to equality, which is crucial in what follows.
We recall that the equality of two arithmetical expressions $e_1$ and $e_2$ is defined 
by the $2^{\text{nd}}$-order encoding 
$$e_1=e_2~~\equiv~~\forall W(W(e_1) \limp W(e_2))$$
Unfolding the definitions of falsity values, we easily get the following lemma:
\begin{lem}\label{lm:equality}
 Given a pole $\pole$, if $e$ is an arithmetical expression, we have 
 
   $$\fv{e_1=e_2}=\begin{cases}
    \fv{\forall X(X\limp X)} & \text{if}~\M\models e_1=e_2 \\
    \fv{\top\limp\bot} &\text{if}~\M\models e_1\neq e_2 \\
  \end{cases}$$
\end{lem}

The following corollaries are straightforward but will be very useful in Sections 5-7, 
so it is worth mentionning them briefly now.

\begin{cor}\label{cor:equalitystack}
 Let $\pole$ be a fixed pole, $e_1,e_2$ some arithmetical expressions, $u\in\Lambda$ a closed term and $\pi\in\Pi$ a stack
 such that $u\cdot\pi\in\fv{e_1=e_2}$. 
 If $\M \models e_1=e_2$ then $u\star\pi\in\pole$\,.
\end{cor}
\begin{proof}
 By Lemma \ref{lm:equality} we have 
 $$u\cdot\pi\in\fv{\forall X(X\limp X)}=\{u\cdot\pi : \exists S\in\Pow(\Pi), \pi\in S \wedge u\in|\dot S|\}$$
 so that $u\in S^{\pole}$ and $u\star\pi\in\pole$
\end{proof}

\begin{cor}\label{cor:equality}
Given a pole $\pole$, if $e_1,e_2$ are arithmetical expressions, and $u\in\Lambda$, $\pi\in\Pi$ 
are such that $u\cdot\pi\notin\fv{e_1=e_2}$, then
 \begin{enumerate}
  \item $\M\vDash e_1=e_2$
  \item $u\star\pi\notin\pole$
 \end{enumerate}
\end{cor}
\begin{proof}
(1)\quad By contraposition: if $\M\vDash e_1\neq e_2$, by Lemma \ref{lm:equality} we have $\fv{e_1=e_2}=\fv{\top\Rightarrow\bot}=\Lambda\times\Pi$, 
 hence $u\cdot\pi\in\fv{e_1=e_2}$. \\
(2)\quad By (1) we have $\fv{e_1=e_2}=\fv{\forall X(X\Rightarrow X)}=\union{S\in\Pow(\Pi)}\fv{\dot S\Rightarrow\dot S}$,
 hence $u\cdot\pi\notin\fv{e_1=e_2}$ implies that if $S=\{\pi\}$, $u\cdot\pi\notin\fv{\dot S \Rightarrow \dot S}$,
 \emph{i.e.} $u\nVdash S$, so $u\star\pi\notin\pole$.
\end{proof}

\section{The specification problem}
\label{s:specification}
\subsection{The specification problem}
\label{ss:specification}
In the continuity of the work done for Peirce's Law~\cite{Peirce}, we are
interested in the specification problem,  
which is to give a purely computational characterization of the
universal realizers of a given formula $A$. 
As mentioned in this paper, this problem is much more subtle than in
the case of intuitionistic realizability, 
what could be justified, amongst other things, by the presence of
extra instructions that  
do not exist in the pure $\lambda$-calculus and by the ability of a
realizer to backtrack at any time.  
Some very simple case, as the identity-type ($\forall X(X\Rightarrow
X)$) or the boolean-type ($\forall X(X\Rightarrow X \Rightarrow X)$), 
are quite easy to specify, but more interestingly, it turns out that
some more complex formul\ae, for instance the Law of Peirce, 
can also be fully specified~\cite{Peirce}.
In the following, we will focus on the generic case of arithmetical formul\ae.
A premise of this work was done by the first author for the particular case of  
 formul\ae~ of the shape $\exr{n}\far{y}(f(x,y)=0)$~\cite{GuiPhD}. 
In the general case (that is with a finite alternation of quantifiers)
an attempt to characterize the threads of universal realizers is also given in an article of Krivine~\cite{Kri03},
but in the end it only provides us with the knowing of the final state,
whereas we are here interested in a specification of the full reduction process.
As in~\cite{GuiPhD}, our method will rely on game-theoretic interpretation of the formul\ae. 
Before going more into details, let us first look at the easiest
example of specification. 

\begin{example}[Identity type]
In the language of second-order logic, the identity type is described
by the formula $\forall X (X\Rightarrow X)$. 
A closed term $t\in\Lambda$ is said to be \emph{identity-like} if
$t\star u \cdot \pi\succ u \star \pi$ for all $u\in\Lambda$ and $\pi\in\Pi$. 
Examples of identity-like terms are of course the identity function $I
\equiv \lambda x . x$, 
but also terms such as $I I$, $\delta I$ 
(where $\delta \equiv \lambda x . xx$), $\lambda x . \cc (\lambda k
. x)$, $\cc (\lambda k . k I \delta k)$, etc. 

\begin{prop}[\cite{Peirce}] For all terms $t \in\Lambda$, the
  following assertions are equivalent: 
\begin{enumerate}
 \item $t \ureal \forall X (X\limp X)$
 \item $t$ is identity-like
\end{enumerate}
\end{prop}
The interesting direction of the proof is $(1)\limp(2)$. 
We prove it with the methods of threads, 
that we use later in Section~\ref{s:subst}.
Assume $t\ureal \forall X (X\limp X)$, and consider
$u\in\Lambda,\pi\in\Pi$. We want to prove that $t\star u \cdot \pi
\eval u\star \pi$. 
We define the pole  
$$ \pole \equiv (\Th(t\star u \cdot \pi))^{c} \equiv
\{p\in\Lambda\star \Pi : (t\star u\cdot\pi\nsucc p)\}$$ 
as well as the falsity value $S=\{\pi\}$.
From the definition of $\pole$, we know that $t\star u \cdot \pi \notin \pole$. 
As $t\real \dot S \limp \dot S$ and $\pi\in\fv{\dot S}$, we get $u\nVdash S$. 
This means that $u\star\pi\notin\pole$, that is $t\star u \cdot \pi
\eval u\star \pi$. 
 
\end{example}

\subsection{Arithmetical formul\ae}
\label{ss:formulae}
In this paper, we want to treat the case of first-order arithmetical
formul\ae, that are $\Sigma^{0}_n$-formul\ae. 
As we explained in Section~\ref{sss:Induction}, 
in order to recover the strength of arithmetical reasoning,
we will relativize all first-order quantifications to the class $\Nat(x)$.
Besides, relativizing the quantifiers make the individuals visible in the stacks: 
indeed, a stack belonging to $\fv{\far{x}A(x)}$ is of the shape $\barr n \cdot\pi$ with
$\pi\in\fv{A(n)}$, 
whereas a stack of $\fv{\forall x A(x)}$ is of the form $\pi\in\fv{A(n)}$
for some $n\in\N$ that the realizers do not have any physical access to.

\begin{definition}\label{def:formulae}
We define inductively the following classes of formul\ae :
\begin{itemize}
 \item $\Sigma^{0}_0$- and $\Pi^0_0$-formul\ae\ are the formul\ae\ of the form
   $f(\vec{e})=0$ where $f$ is a primitive recursive function and 
   $\vec{e}$ a list of first-order expressions.
  \item $\Pi^{0}_{n+1}$-formul\ae\ are the formul\ae\ of the form $\far x F$,
 where $F$ is a $\Sigma^{0}_{n}$-formula.
 \item $\Sigma^{0}_{n+1}$-formul\ae\ are the formul\ae\ of the form $\exr x F$,
 where $F$ is a $\Pi^{0}_{n}$-formula.
\end{itemize}
\end{definition}

In the ground model $\M$, any closed $\Sigma^{0}_n$- or $\Pi^{0}_n$-formula $\Phi$ naturally  
induces a game between two players $\exists$ and $\forall$, that we
shall name \elo\ and \abe\ from now on. 
Both players instantiate the corresponding quantifiers in turns, 
\elo\ for defending the formula and \abe\ for attacking it. 
The game, whose depth is bounded by the number of quantifications,  proceeds as follows:
\begin{itemize}
\item When $\Phi$ is $\exists x \Phi'$, \elo\ has to give an integer
  $m\in\N$, and the game goes on over the closed formula $\Phi'\{x:=m\}$. 
\item When $\Phi$ is $\forall y \Phi'$, \abe\ has to give an integer
  $n\in\N$, and the game goes on over the closed formula $\Phi'\{y:=n\}$. 
\item When $\Phi$ is atomic and $\M\vDash\Phi$ ($\Phi$ is true), \elo\
  wins, otherwise \abe\ wins. 
\end{itemize}

We say that a player has a \emph{winning strategy} if (s)he has a way of
playing that ensures him/her the victory independently of the opponent
moves. 
It is obvious from Tarski's definition of truth that a closed arithmetical formula $\Phi$ is true in
the ground model if and only if \elo\ has a winning strategy.

The problem with this too simple definition is that there exists true formul\ae~
whose game only has non-computable winning strategies (as we shall see below), so that they cannot be
implemented by $\lambda$-terms.
This is why in classical logic, we will need to relax the rules of the above game 
to allow backtracking.

\subsection{The Halting problem or the need of backtrack}
\newcommand{\mach}{\mathscr{M}}
\label{ss:halt}
\newcommand{\Halt}{\mathrm{\mathop{ Halt}}}
For instance, let us consider one of the primitive recursive functions $f:\N^{3}\to\N$ such that
\begin{center}
\begin{tabular}{ccc}
$f(m,n,p)=0$ & \hspace{1cm}~iff~  & $(n>0 \wedge \Halt(m,n)) \vee
  (n=0\wedge \neg \Halt(m,p))$ 
\end{tabular}
\end{center}
where $\Halt(m,n)$ is the primitive recursive predicate  
expressing that the $m^{\mathrm{th}}$ Turing machine
has stopped before $n$ evaluation steps (in front of the empty tape).
From this we consider the game on the formula 
$$\Phi_H\equiv \far{x} \exr y \far z (f(x,y,z)=0)$$
that expresses that any Turing machine terminates or does not terminate.
(Intuitively $y$ equals $0$ when the machine $x$ does not halt, 
and it represents a number larger than the execution length of $x$ otherwise.)
Yet, there is no pure $\lambda$-term that can compute directly from an $m\in\N$ an integer $n_m$ such that
$\far z (f(m,n_m,z)=0)$ (such a term would break the halting
problem). 
However, $\Phi_H$ could be classically realized, using the $\cc$ instruction.
Let $\Theta$ be a $\lambda$-term such that :
$$\Theta \star \barr m\cdot \barr n \cdot t_0 \cdot t_1 \cdot \pi \eval 
\left\{\begin{array}{ll} 
        t_0 \star \pi & \textrm{if the $m^{\mathrm{th}}$ Turing
          machine stops before $n$ steps}\\ 
        t_1 \star \pi & \textrm{otherwise}
       \end{array}\right.$$
and let $t_H$ be the following term :
$$\begin{array}{r@{~\equiv~}l} 
 T[m,u,k] & \lambda p v.\Theta~m~p~(k~(u~p~\lambda pv.v))~v\\
t_H &  \lambda mu. \cc ~(\lambda k.u~\barr 0~ T[m,u,k])
\end{array}$$

If we think of $t_H$ as a strategy for \elo\ \emph{with backtrack allowed},
we can analyze its computational behaviour this way:
\begin{itemize}
 \item First \elo\ receives the code $m$ of a Turing machine $\mach$, 
 and chooses to play $n=0$, that is "\emph{$\mach$ never stops}".
 \item Then \abe\ answers a given number of steps $p$, and \elo\ checks
   if $\mach$ stops before $p$ steps and distinguishes two cases : 
 \begin{itemize}
 \item either $\mach$ is still running after $p$ steps, hence
   $f(m,0,p)=0$ and \elo\ wins. 
 \item either $\mach$ does stop before $p$ steps, then \elo\ backtracks
   to the previous position and instead of $0$, it plays $p$,  
 that is "\emph{$\mach$ stops before $p$ steps}", which ensures him
 victory whatever \abe\ plays after. 
 \end{itemize}
\end{itemize}
\begin{prop}
 $t_H\ureal \Phi_H$
\end{prop}

\begin{proof}
 Let us consider a fixed pole $\pole$ and let $m\in\N$ be an integer,
 $\mach$ be the $m^{\mathrm{th}}$ Turing machine,  
 and a stack $u\cdot\pi\in\fv{\exr y \far z (f(m,y,z)=0)}$, and let us
 prove that $t_H\star \barr m \cdot u \cdot \pi\in\pole$. 
 We know that 
 $$t_H\star \barr m \cdot u \cdot \pi \eval u\star \barr 0 \cdot
 T[\barr m,u,\k_\pi]\cdot \pi$$ 
 by anti-reduction, it suffices to prove that $T[\barr
   m,u,\k_\pi]\real \far z (f(m,0,z)=0)$. 
 Thus let us consider $p\in\N$ and a stack $u'\cdot\pi'\in\fv{f(m,0,p)=0}$.
 We distinguish two cases:
 \begin{itemize}
  \item $\mach$ is still running after $p$ steps (that is
    $\M\vDash\neg\Halt(m,p)$).  
  In this case, we have $f(m,0,p)=0$, and so by Corollary
  \ref{cor:equalitystack}, $u'\star\pi'\in\pole$. 
  Furthermore, by definition of $\Theta$, we have 
  $$T[\barr m,u,\k_\pi] \star \barr p \cdot u' \cdot \pi'\eval u'\star
  \pi'\in\pole $$ 
  which concludes the case by anti-reduction.
  
  \item $\mach$ stops before $p$ steps ($\M\vDash \Halt(m,p)$). By
    definition of $\Theta$, we have in this case  
  $$T[\barr m,u,\k_\pi] \star \barr p \cdot u' \cdot \pi'\eval \k_\pi
    \star (u~\barr p~(\lambda pv.v))\cdot \pi'  
  \eval u\star \barr p\cdot \lambda pv.v\cdot \pi$$
  hence it suffices to show that $\lambda pv.v \real \far z (f(m,p,z)=0)$. 
  But this is clear, as $\M\vDash \Halt(m,p)$, we have for any
  $s\in\N$, $\M\vDash f(m,p,s)=0$.  
  Therefore if we consider any integer $s\in\N$ and any stack
  $u''\cdot\pi''\in\fv{f(m,p,s)=0}$,  
  as in the previous case, from Corollary \ref{cor:equalitystack} we
  get $u''\star\pi'' \in\pole$ and 
  $$\lambda pv.v\star \barr s \cdot u''\cdot \pi''\eval
  u''\star\pi''\in\pole \eqno \qedhere$$ 
 \end{itemize}
\end{proof}

This leads us to define a new notion of game with backtrack over
arithmetical formul\ae. 

\subsection{$\G^{0}_{\Phi}$: a first game with backtrack}
\label{ss:game}

From now on, to simplify our work, we will always consider $\Sigma^{0}_{2h}$-formul\ae, that is of the form:
$$\exr x_1 \far y_1 \ldots \exr x_h \far y_h f(\vc{x}{h},\vc{y}{h})=0$$
where $h\in\N$ and the notation $\vc{x}{i}$ refers to the tuple $(x_1,\ldots,x_i)$ 
(we will denote the concatenation by $\cdot$ : $\vc{x}{i}\cdot x_{i+1}=\vc{x}{i+1}$).
It is clear that any arithmetical formul\ae~ can be written equivalently in that way, adding some useless quantifiers if needed.

Given such a formula $\Phi$, we define a game $\G^{0}_{\Phi}$ between \elo\ and \abe\ whose rules  are basically the same as they were before, 
except that we will keep track of all the former $\exists$-positions, allowing \elo\ to backtrack. 
This corresponds to the definition of Coquand's game~\cite{Coquand95}.
We call an \emph{$\exists$-position} of size $i\in\llbracket 0,h\rrbracket$ a pair of tuple of integers $(\vc{m}{i},\vc{n}{i})$
standing for the instantiation of the variables $\vc{x}{i},\vc{y}{i}$,
while a $\forall$-position will be a pair of the form
$(\vc{m}{i+1},\vc{n}{i})$.  
We call \emph{history of a game} and note $\H$ the set of every former
$\exists$-positions. 
The game starts with an empty history ($\H=\{\emptyset\}$) and
proceeds as follows:  
\begin{itemize}
\item $\exists$-move: \elo\ chooses a position
  $(\vc{m}{i},\vc{n}{i})\in\H$ for some $i\in\llbracket
  0,h-1\rrbracket$, and proposes $m_{i+1}\in\N$, 
so that $(\vc{m}{i+1},\vc{n}{i})$ becomes the current $\forall$-position.
\item $\forall$-move: \abe\ has to answer with some $n_{i+1}\in\N$ to
  complete the position. 
\end{itemize}
If $i+1=h$ and $f(\vc{m}{h},\vc{n}{h})=0$, then \elo\ wins and the game stops. 
Otherwise, we simply add the new $\exists$-position
$(\vc{m}{i+1},\vc{n}{i+1})$ to $\H$, 
and the game goes on. 
We say that \abe\ wins if the game goes on infinitely, that is if \elo\
never wins. 

Given a set $\H$ of former $\exists$-positions, we will say that \elo\
has a \emph{winning strategy} and write $\H\in\W^{0}_\Phi$ if 
she has a way of playing that ensures her a victory, independently of
future \abe\ moves.  

Formally, we define the set $\W^{0}_\Phi$ by induction with the
two following rules: 

\begin{enumerate}
 \item If there exists $(\vc{m}{h},\vc{n}{h})\in\H$ such that $\M
   \vDash f(\vc{m}{h},\vc{n}{h})=0$: 
    \begin{prooftree}
       \AxiomC{}
       \RightLabel{\scriptsize(Win)}
       \UnaryInfC{$\H\in\W^0_{\Phi}$}
     \end{prooftree}
   
 \item For all $i<h$, $(\vc{m}{i},\vc{n}{i})\in\H$ and
   $m\in\N$
      \begin{prooftree}
     \AxiomC{$\H\cup \{(\vc{m}{i}\cdot
       m,\vc{n}{i}\cdot n)\} \in \W^{0}_\Phi \quad \forall n\in\N$}
     \RightLabel{\scriptsize(Play)}
     \UnaryInfC{$\H \in \W^{0}_{\Phi}$}
   \end{prooftree}
   
\end{enumerate}

Given a formula $\Phi$, the only difference between this game and the
one we defined in Section \ref{ss:formulae} 
is that this one allows \elo\ to make some wrong tries before moving to
a final position.  
Clearly, there is a winning strategy for $\G^{0}_\Phi$ if and only if
there was one in the previous  
game\footnote{It suffices to remove the "bad tries" to keep only the
  winning move}. 
It is even easy to see that for any formula $\Phi$, we have 
\begin{prop}\label{prop:mdl}
 $\M\vDash\Phi$ iff $\{\emptyset\}\in\W^{0}_\Phi$
\end{prop}

Given a formula $\Phi$, in both games the existence of a winning strategy is equivalent 
to the truth in the model, hence such a definition does not carry anything new from an outlook of model theory,
the interest of this definition is fundamentally computational. For instance, for the halting problem, this will now allow \elo\ to use the
strategy we described in the previous section. 

Besides, it is worth noting that in general, the match somehow grows
among a tree of height $h$, 
as we shall see in the following example.

\begin{example}\label{ex:thread}
\newcommand{\dotminus}{{\raisebox{0.8ex}{.}\hspace{-1ex}-}}
We define the following function 
$$g:\left\{\begin{array}{ccl}
           \N^{2}&\to&\N \\
           (x,y)&\mapsto & x+(1~ \dotminus x)y
\end{array}\right.$$
where $\dotminus$ refers to the truncated subtraction.
Notice that $g(x,\cdot)$ is clearly bounded if $x\neq 0$.
Then we consider $f$ a function such that 
$$f(x_1,y_1,x_2,y_2)=0  \text{~~if and only if~~} (x_1=y_1 \vee
g(x_1,x_2)> g(y_1,y_2))$$ 
Finally, we define the formula $\varphi\in\Sigma^{0}_{4}$ 
$$\varphi \equiv \exr{x_1}\far y_1 \exr x_2 \far y_2 (f(x_1,y_1,x_2,y_2)=0)$$
which expresses that there exists $x_1$ (in fact 0) such that
$g(y_1,\cdot): z\mapsto g(y_1,z) $ is bounded for every $y_1\neq
x_1$. 
The shortest strategy for \elo\ to win that game would be to give 0 for
$x_1$, wait for an answer $m$ for $y_1$,  
and give $m+1$ for $x_2$. But we can also imagine that \elo\ might try
0 first, receive \abe\ answer, 
and then change her mind, start from the beginning with $1$, try
several possibilities before going back to the winning position. 
If we observe the positions \elo\ will reach for such a match, we
remark it draws a tree (see Figure \ref{fig:tree}). 
\newcommand{\xvec}[1]{
#1}
\begin{figure}[htp]\label{fig:tree}
\begin{tabular}{c|c}
\hline\hline\noalign{\medskip} 
\begin{tabular}{c|c|c|c}
Start & \elo\ move & \abe\ & new $\exists$-position \\ \hline
$\emptyset,\emptyset$ &  0 & 1 & $\xvec{0},\xvec{1}$ \\
$\emptyset,\emptyset$ &  1 & 0 & $\xvec{1},\xvec{0}$ \\
$\xvec{1},\xvec{0}$ &  1 & 1 &
$\xvec{1{\cdot}1},\xvec{0{\cdot}1}$ \\ 
$\xvec{1},\xvec{0}$ &  2 & 2 & $\xvec{1{\cdot}2},
\xvec{0{\cdot}2}$ \\ 
$\emptyset,\emptyset$ &  2 & 0 & $\xvec{2},\xvec{0}$ \\
$\xvec{0},\xvec{1}$ &  2 & 1 & $\xvec{0{\cdot}2},
\xvec{1{\cdot}1}$ \\ 
$\xvec{0{\cdot}2},\xvec{1{\cdot}1}$ &
\elo\ wins & / & / \\ 
\end{tabular}
&
\begin{minipage}{0.325\textwidth}
\scalebox{0.8}{\begin{tikzpicture}[level/.style={sibling
        distance=20mm/#1}]{0pt}{0pt}{2pt}{40pt} 
\vspace{2cm}
\node [rectangle] (r){$\emptyset,\emptyset$}
child {node [rectangle] (b) {$\xvec{0},\xvec{1}$} 
  child {node [rectangle,draw] (e)
    {$\xvec{0{\cdot}2},\xvec{1{\cdot}1}$}}}
child {node [rectangle] (a) {$\xvec{1},\xvec{0}$}
  child {node [rectangle] (b)
    {$\xvec{1{\cdot}1},\xvec{0{\cdot}1}\quad$}} 
  child {node [rectangle] (d)
    {$\quad\xvec{1{\cdot}2},\xvec{0{\cdot}2}$}}
}
child {node [rectangle] (a) {$\xvec{2},\xvec{0}$}
}
  ;
\end{tikzpicture}}
\end{minipage}\\
\noalign{\medbreak}\hline\hline
\end{tabular}
\caption{Example of a match for $\G^{0}_{\varphi}$}
\end{figure}
We shall formalize this remark later, but we strongly advise the reader to keep this representation in mind all along the next section.
 \end{example}

\section{Implementing the game}
\label{s:subst}
\subsection{Substitutive Game: $\G^{1}_\Phi$}
\label{ss:g1}
Now that we have at our disposal a notion of game that seems to be suitable to
capture computational content of classical theorems,  
we shall  adapt it to play with realizers. 
Considering a formula 
$\Phi\equiv \exr x_1\far y_1\dots \exr x_h\forall y_h (f(\vec{x}_h, \vec{y}_h)=0)$ 
 we will have to consider sub-formul\ae\ of $\Phi$
to write down proofs about $\Phi$.  
Therefore we give the following abbreviations
that we will use a lot in the following:
$$\begin{array}{r@{~\equiv~}l}
 E_i & \forall X_{i+1} (A_{i+1} \Rightarrow X_{i+1})\\ 
 A_i & \far{x_i} (\far{y_i}E_i \Rightarrow X_i)\\   
 E_{h} &\forall W
 (W(f(\vc{x}{h},\vc{y}{h})) 
 \Rightarrow W(0)) \\ 
\end{array}\eqno\begin{array}{r}
(\forall i\in\Int{0,h-1})	\\ 
(\forall i\in\Int{1, h})\\
\\
\end{array}$$
One can easily check that $E_0\equiv\Phi$ 
and that the other definitions correspond to the unfolding of the quantifiers.

  In order to play using realizers, we will slightly change the setting of
$\G^{0}_{\Phi}$, adding processes. 
One should notice that we only add more information, so that the game
$\G^{1}_\Phi$ is somehow a ``decorated" version of $\G^{0}_\Phi$. 

  To describe the match, we use $\exists$-positions --which
  are just processes-- and 
  $\forall$-positions --which are $4$-uples of the shape
  $(\vc{m}{i}, \vc{n}{i}, u, \pi)\in\N^{\leq h}\times\N^{\leq h}
  \times\Lambda_c\times\Pi$. If $i=h$, we say that the move is
  \emph{final} or \emph{complete}. In a given time $j$, the set of all
  $\forall$-positions reached before is called \emph{the history} and is
  denoted as $\H_j$. At each time $j$, the couple given by the current
  $\exists$-position $p_j$ and the history $\H_j$ is called the $j$-th
  state. The state evolves throughout the match according to the
  following rules:
  \begin{enumerate}
  \item \elo\ proposes a term
    $t_0\in \Pl$ supposed to defend $\Phi$ and \abe\ proposes a stack
    $u_0\cdot \pi_0$ supposed to attack the formula $\Phi$. We say
    that at time $0$, the
    process $p_0:=t_0\star u_0\cdot \pi_0$ is the current
    $\exists$-position and $\H_0:=\{(\emptyset,\emptyset,u_0,\pi_0)\}$ is the current history. This
        step defines the initial state $\langle p_0, \H_0\rangle$.
  \item Assume $\langle p_j, \H_j\rangle$ is the $j^{\text{th}}$
        state. Starting from $p_j$ \elo\ evaluates $p_j$ in order to reach
    one of the following situations:
    \begin{itemize}
      \item $p_j\succ u\star\pi$ for some (final) $\forall$-position
        $(\vc{m}{h}, \vc{n}{h}, u, \pi)\in \H_j$. In this case, \elo\
        wins if $\M\models f(\vc{m}{h}, \vc{n}{h})=0$. 
      \item $p_j\succ u\star \overline{m}\cdot t\cdot \pi$ for
        some (not final) $\forall$-position $(\vc{m}{i}, \vc{n}{i}, u,
        \pi)\in\H_j$ where $i<h$. If so, \elo\ \emph{can} decide to play by
        communicating her answer $(t,m)$ to \abe\ and standing for his
        answer, and \abe\ \emph{must} answer a new integer $n$ together
        with a new stack $u'\cdot\pi'$. The $\exists$-position 
        becomes $p_{j+1}:=t\star\overline{n}\cdot u'\cdot \pi'$ 
        and we add the $\forall$-position to the history:
        $\H_{j+1}:=\H_j\cup \{(\vc{m}{i}\cdot m,
        \vc{n}{i}\cdot n, u', \pi')\}$. This step defines the next
      state $\langle p_{j+1},\H_{j+1}\rangle$   
    \end{itemize}
    If none of the above moves is possible, then \abe\ wins. 
  \end{enumerate}
Intuitively, a state $\langle p, \H\rangle$ is \emph{winning} for
\elo\ if and only if she \emph{can play} in such a way that \abe\ \emph{will lose anyway}, 
independently of the way he might play.

Start with a term $t$ is a ``good move'' for \elo\ if and only if,
proposed as a defender of the formula, $t$ defines an initial winning state (for \elo), 
independently from the initial stack proposed by \abe. In this case, adopting the
point of view of \elo, we just say
that $t$ is a \emph{winning strategy} for the formula $\Phi$.

Since our characterization of realizers will be in terms of
  winning strategies, we might formalize this notion. We define
  inductively the set of \emph{winning states} --which is a syntactic
  object-- by means of a deductive system:

\begin{itemize}
 \item if $\exists (\vc{m}{h},\vc{n}{h},u,\pi)\in\H$ s.t.  $p\eval
   u\star\pi$ and $\M \vDash f(\vc{m}{h},\vc{n}{h})=0$ : 
 $$\infer[{\mbox{\scriptsize (Win)}}]{\pair{p,\H} \in \W^{1}_\Phi}{}$$
 \vspace{-0.3cm}
 \item for every $(\vc{m}{i},\vc{n}{i},u,\pi)\in\H$, $m\in\N$ s.t.
   $p\eval u\star \overline{m} \cdot t \cdot \pi$ : 
 $$\infer[{\mbox{\scriptsize (Play)}}]{\pair{p,\H} \in \W^{1}_{\Phi}}
 {\pair{t\star \overline{n} \cdot u' \cdot \pi',
     \H\cup\{(\vc{m}{i}\cdot m,\vc{n}{i}\cdot n,u',\pi')\}} \in
   \W^{1}_\Phi  
& \forall (n',u',\pi')\in\N\times\Lambda\times\Pi}$$
\end{itemize}

A term $t$ is said to be a winning strategy for $\Phi$ if for any
handle $(u,\pi)\in\Lambda\times\Pi$,  
we have $\pair{t\star u\cdot\pi, \{(\emptyset,\emptyset,u,\pi)\}}\in\W^{1}_\Phi$.

\begin{prop}[Adequacy]\label{prop:adequacy1}
 If $t$ is a winning strategy for $\mathds{G}^{1}_\Phi$, then $t\ureal \Phi$
\end{prop}
\begin{proof}
 We will see a more general game in the following section for which we will prove the adequacy property (Proposition \ref{prop:adequacy}) 
 and which admits any winning strategy of this game as a winning strategy (Proposition \ref{prop:inclusion}), 
 thus proving the adequacy in the current case. Furthermore, the proof we give for Proposition \ref{prop:adequacy} is suitable for this game too.
\end{proof}

\subsection{Completeness of $\G^{1}_\Phi$ in presence of interaction constants}

In this section we will show the completeness of $\G^{1}_\Phi$ by substitution over the thread of execution of a universal realizer of $\Phi$.
As observed in section \ref{ss:game}, the successive $\exists$-positions form a tree. 
We give thereafter a formal statement for this observation, which will allow us to prove the completeness of this game.
We shall now give a formal definition of a tree.

\begin{definition}
 A (finite) \emph{tree} $\T$ is a (finite) subset\footnote{Observe that $|\T|$ (the cardinality of $\T$) 
 coincides with the usual definition of the size of $\T$.}
 of $\N^{<\omega}$ such that if $\tau\cdot c\in\T$ and $c\in\N$, then
 $\tau\in\T$ and $\forall c'<c, \tau\cdot c'\in\T$, where the $\cdot$ operator denotes the concatenation.
 If $\tau=c_0\cdots c_k$, we use the notation $\tau_{|i}=c_0\cdots c_i$, 
 and we note $\tau\sqsubset\sigma $ ( $\sigma$ \emph{extends} $\tau$) when :
 $$\tau\sqsubset \sigma ~\equiv~ \sigma_{|k}=c_0\cdots c_k=\tau$$

 We call \emph{characteristic function} of a tree $\T$ any partial function $\varphi:\N\to\P(\N^{<\omega})$ such that:
 \begin{enumerate}
 \item $\forall n\in\dom(\varphi), \{\varphi(m) : m\le n\}$ is a tree
 \item $\varphi(|\T|)=\T$
 \end{enumerate}
\end{definition}

\begin{lem}\label{lm:scheme}
 Assume the calculus of realizers is deterministic, and let $t_0$ be a universal realizer of $\Phi\in\Sigma^{0}_{2h}$. 
 \mauricio{Consider} $(n_j)_{j\in\N}$ an infinite sequence of integers, 
 $(\kappa_j )_{j\in\N}$ an infinite
sequence of (pairwise distinct) \mauricio{interaction constants} that do not occur in $t_0$ and if
$(\alpha_j )_{j\in\N}$ is an infinite sequence of \mauricio{substitutive and non-generative} stack constants. 
Then there exists two integers $f,s\in\N$,
two finite sequences $t_0,\ldots,t_f\in\Lambda$ and $m_1,\ldots,m_f\in\N$
as well as a tree characteristic function $\varphi:\Int{0,f}\to\N^{<\omega}$
such that:

\begin{center}
$\begin{array}{lr @{~\eval~}l @{\hspace{0.6cm}} r}

  &t_{0} \star \kappa_{0}\cdot\alpha_{0} & \kappa_{0}\star\overline{m}_1\cdot t_{1}\cdot\alpha_{0} & \\
\forall i\in\Int{1,f-1} & t_i\star \overline{n_i}\cdot\kappa_i\cdot\alpha_i & \kappa_{j}\star \overline{m_{i+1}}\cdot t_{i+1}\cdot \alpha_{j} 
& \left(\begin{array}{l}\textrm{with } j\le i\\ 
\varphi(j)\sqsubset \varphi(i+1)\end{array}\right) \\ 
  &t_f\star \overline{n_f}\cdot\kappa_f\cdot\alpha_f & \kappa_s \star \alpha_s & \\
\end{array}$
\end{center}
where $ |\varphi(s)|=h$ and  $\M \vDash f(\vec{m}_{\varphi(s)},\vec{n}_{\varphi(s)})=0$
\end{lem}

\begin{example}
Before doing the proof, let us have a look at an example of such a thread scheme for a formula $\Phi\in\Sigma^{0}_{4}$ 
(as we considered in Example \ref{ex:thread}) and to the corresponding tree and characteristic function.
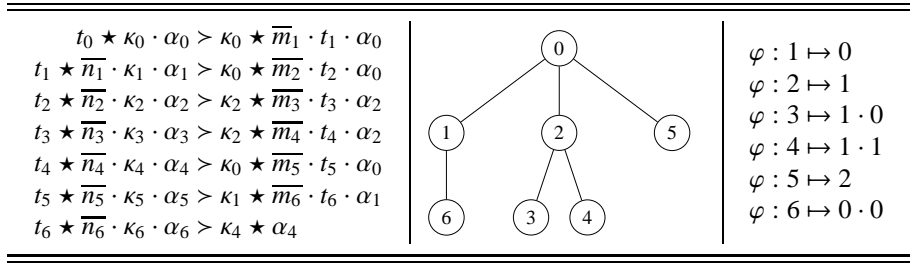
\begin{figure}[htp]
\begin{tabular}{c|c|c}
\hline\hline\noalign{\medskip}
$\begin{array}{r @{~\eval~}l}
t_{0} \star \kappa_{0}\cdot\alpha_{0} & \kappa_{0}\star\overline{m}_1\cdot t_{1}\cdot\alpha_{0}  \\               
t_1\star \overline{n_1}\cdot\kappa_1\cdot\alpha_1 & \kappa_{0}\star \overline{m_{2}}\cdot t_{2}\cdot \alpha_{0} \\
t_2\star \overline{n_2}\cdot\kappa_2\cdot\alpha_2 & \kappa_{2}\star \overline{m_{3}}\cdot t_{3}\cdot \alpha_{2} \\
t_3\star \overline{n_3}\cdot\kappa_3\cdot\alpha_3 & \kappa_{2}\star \overline{m_{4}}\cdot t_{4}\cdot \alpha_{2} \\
t_4\star \overline{n_4}\cdot\kappa_4\cdot\alpha_4 & \kappa_{0}\star \overline{m_{5}}\cdot t_{5}\cdot \alpha_{0} \\
t_5\star \overline{n_5}\cdot\kappa_5\cdot\alpha_5 & \kappa_{1}\star \overline{m_{6}}\cdot t_{6}\cdot \alpha_{1} \\
t_6\star \overline{n_6}\cdot\kappa_6\cdot\alpha_6 & \kappa_4 \star \alpha_4  
\end{array}$
&\begin{minipage}{0.3\textwidth}
\scalebox{0.75}{
\begin{tikzpicture}[level/.style={sibling distance=20mm/#1}]{0pt}{0pt}{2pt}{40pt}
\vspace{2cm}
\node [circle,draw] (r){$0$}
  child {node [circle,draw] (a) {$1$}
    child {node [circle,draw] (b) {$6$}
        } 
    }
  child {node [circle,draw] (b) {$2$} 
            child {node [circle,draw] (d) {$3$}}
            child {node [circle,draw] (e) {$4$}}
      } 
  child {node [circle,draw] (a) {$5$}
      };
\end{tikzpicture}}
\end{minipage}
&
$\begin{array}{c@{~\mapsto~}l}
  \varphi : 1 & 0\\
  \varphi : 2 & 1 \\
  \varphi : 3 & 1\cdot0 \\
  \varphi : 4 & 1\cdot1 \\
  \varphi : 5 & 2 \\
  \varphi : 6 & 0\cdot0 \\
 \end{array}$\\
\noalign{\medbreak}\hline\hline
\end{tabular}
\caption{A thread scheme for $\Phi\in\Sigma^{0}_{4}$}\label{fig:scheme}
\end{figure}
\newcommand{\var}{\mathop{\mathrm{var}}}
\newcommand{\vvar}{\overset{\longrightarrow}{\mathop{\mathrm{var}}}}
\end{example}

We observe that we could actually labeled any node of the tree using its order 
of apparition in the enumeration of $\T$ with $\varphi$.
\begin{definition}
Given such a thread scheme and a path $\tau\in\T$, we define
$m_{\tau}=m_{\varphi^{-1}(\tau)}$ (integer $m$ at the node $\tau$),
$\vec{m}_\tau=(m_{\tau_{|1}},m_{\tau_{|2}},\ldots,m_\tau)$ (integers
$m$ along the path) and the substitution along $\tau$ is : 
$$\sigma(\tau)=\{x_i:=m_{\tau_{|i}}\}^{|\tau|}_{i=1}\{y_i:=n_{\tau_{|i}}\}^{|\tau|}_{i=1}$$
\end{definition}
For instance, in Figure \ref{fig:scheme}, for $\tau=1\cdot1$ 
(wich corresponds to the choosen final position
$\kappa_4\star\alpha_4$), we have : 
$$\sigma(\tau)\equiv\{x_1:=m_2,x_2:=m_4,y_1:=n_2,y_2:=n_4\}$$

\begin{proof}[Proof of Lemma \ref{lm:scheme}]
We build a sequence  $(Q_i)_{i\in\N}$ of sets of processes and a sequence 
of characteristic functions $(\varphi_i)_{i\in\N}$ for some trees $(\T_i)_{i\in\N}$,
such that at each step $i\in\N$, $Q_i$ is either empty either of the
form $\Th(p)$  
for some $p\in\Lambda\times\Pi$ :
\begin{itemize}
 \item $i=0$ : we set $Q_0=\Th(t_0\star\kappa_0\cdot\alpha_0)$ and
   $\varphi_0:0\mapsto \emptyset$ 
 \item $i\in\N$ : given $Q_i$ and $\varphi_i$, 
 if there exist\footnote{Note that as the calculus is deterministic and the constants $\kappa_j$ inert, if such $j, m_{i+1},t_{i+1}$ exist, they are unique}
 $j\in\N, m_{i+1}\in\N$ and $t_{i+1}\in\Lambda$ such that 
 $\kappa_j\star\overline{m_{i+1}}\cdot t_{i+1} \cdot \alpha_j \in Q_i$ 
  we set:
 \begin{center}
 $Q_{i+1}:=\Th(t_i\star
   \overline{n_{i+1}}\cdot\kappa_{i+1}\cdot\alpha_{i+1})$ \hspace{1cm} 
 $\varphi_{i+1}:=\left\{\begin{array}{ccl} k\le i & \mapsto & \varphi_i(k) \\
 i+1 &\mapsto& \varphi_i(j)\cdot c \end{array}\right.$
 \end{center}
 where $c:=\min\{n\in\N\mid \varphi_i(j)\cdot n \notin \T_i\}$. It is
 easy to check  
 that if $\varphi_i$ is a characteristic function for $\T_i$, then
 so is $\varphi_{i+1}$ for $\T_i\cup\{\varphi_i(j)\cdot c\}$;
\end{itemize}
otherwise $Q_{i+1}:=\emptyset$ and $\varphi_{i+1}:=\varphi_i$. 
%
We define $Q_\infty:=\bigcup_{i\in\N}Q_{i}$, $\pole:=Q_\infty^{c}$ and
$\varphi:=\lim_{i\in\omega}\varphi_i$. 
We prove by induction that for any $0\le i \le h$, the following
statement holds: 
$$\exists j\in\N, |\varphi(j)|=i \textrm{ such that }
\kappa_j\cdot\alpha_j\notin\fv{E_i[\sigma(\varphi(j))]} 
\eqno{(\text{IH}_i)}$$


\begin{paragraph}{\bf IH$_0$:}
From the definition of $\pole$, we have
$t_0\star\kappa_0\cdot\alpha_0\notin\pole$.  
Besides, we know that $t_0\real E_0$,
so that $\kappa_0\cdot\alpha_0\notin\fv{E_0}$.
\end{paragraph}
\begin{paragraph}{\bf IH$_{i+1}$:}
Assume we have IH$_i$, for $0\le i<h$, that is 
$\exists {j_i}\in\N, |\varphi({j_i})|=i$ such that 
$$\kappa_{j_i}\cdot\alpha_{j_i}\notin\fv{E_i[\sigma(\varphi(j_i))]}$$	
Recall that $E_i=\forall X_{i+1} (A_{i+1} \Rightarrow X_{i+1})$,
hence $\kappa_{j_i}\nVdash
A_{i+1}[X_{i+1}:=\dot\alpha_\tau][\sigma(\varphi_i({j_i})]$. 
Therefore there exists $m\in\N$ and
$t\real\far{y_{i+1}}E_{i+1}[\sigma(\varphi({j_i}))]\{x_{i+1}:=m\}$
such that 
$\kappa_{j_i}\star \overline{m}\cdot t \cdot \alpha_{j_i}\notin\pole$. 
By definition of $\pole$, it means that there is some $j\in\N$ such 
that this process belong to $Q_{j}$, so that by definition of $Q_{j+1}$
we have $t_{j+1}=t,m_{j+1}=m$, $\varphi(j+1)_{|i}=\varphi(j_i)$,
$$t_{j+1}\star \overline{n_{j+1}} \cdot \kappa_{j+1} \cdot \alpha_{j+1}\notin\pole$$
Using the fact that $t_{j+1}\real\far{y_{i+1}}E_{i+1}[\sigma(\varphi(j_i))]\{x_{i+1}:=m\}$, we finally get that 
$$\kappa_{i+1} \cdot \alpha_{i+1}\notin \fv{E_{i+1}[\sigma(\varphi(j+1))]}$$
since $\sigma(\varphi(j+1))=\sigma(\varphi(j_i))\{x_{i+1}:=m_{j+1};y_{i+1}:=n_{j+1}\}$.
\end{paragraph}

We obtain then for IH$_h$ the following statement :
$$\exists s\in\N, |\varphi(s)|=h 
\textrm{ such that } \kappa_s\cdot\alpha_s\notin\fv{\forall W(f(\vc{m}{\varphi(s)},\vec{n}_{\varphi(s)}))\Rightarrow W(0)}$$
Applying the lemma \ref{cor:equality}, we get that $\M\vDash f(\vec{m_\sigma},\vec{n_\sigma})=0$ and $\kappa_s\star\alpha_s\notin\pole$.
Hence there exists $f\in\N$ such that $\kappa_s\star\alpha_s\in Q_f$, thus 
$$t_f\star \overline{n_f}\cdot\kappa_f\cdot\alpha_f \eval \kappa_s \star \alpha_s, 
\textrm{ with }\M \vDash f(\vec{m}_\sigma,\vec{n}_\sigma)=0$$
that is the last line of the expected thread scheme. 

Besides, by definition of $Q_f$ and $\varphi_f$, we clearly have that for any $i\in\Int{0,f-1}$,
there exists $j\in\N$ such that $j\le i$ and 
$$t_i\star \overline{n_i}\cdot\kappa_i\cdot\alpha_i \eval \kappa_{j}\star \overline{m_{i+1}}\cdot t_{i+1}\cdot \alpha_{j}\eqno \qedhere$$
\end{proof}

Note that, as the constants $\kappa_i$ and $\alpha_i$ are substitutive, the function $\varphi$ and the integers $f$ and $s$ only 
depend on the sequence $(n_i)_{i\in\N}$. In other words, the threads scheme is entirely defined by this sequence.

\begin{prop}[Completeness of $\mathds{G}^{1}_\Phi$ in presence of interaction constants]
If the calculus of realizers is deterministic and contains infinitely many interaction constants as well as infinitely
many substitutive and non generative stack constants, then every universal realizer of
an arithmetical formula $\Phi\in\Sigma^{0}_{h}$ is a winning strategy for the game $\mathds{G}^{1}_\Phi$
\end{prop}

\begin{proof}
Consider $\Phi\in\Sigma^{0}_{h}$ and a closed term $t_0\ureal \Phi$. 
Given any infinite sequence of (pairwise distinct) non generative constants $(\kappa_i )_{i\in\N}$ 
that do not occur in $t_0$ and any sequence of stack constants $(\alpha_i )_{i\in\N}$, 
we have shown that for any sequence $(n_i)_{i\in\N}$ of integers, there exists 
two integers $f,s\in\N$,
two finite sequences of integers $m_0,\ldots,m_f\in\N$ and closed terms $t_0,\ldots,t_f\in\Lambda$ and 
a finite tree $\T$ whose characteristic function $\varphi$ verifies $|\varphi(s)|=h$: 
\begin{center}
$\begin{array}{lr @{~\eval~}l @{\hspace{0.4cm}} r}

  &t_{0} \star \kappa_{0}\cdot\alpha_{0} & \kappa_{0}\star\overline{m}_1\cdot t_{1}\cdot\alpha_{0} & \\
\forall i\in\Int{1,f-1}
& t_i\star \overline{n_i}\cdot\kappa_i\cdot\alpha_i 
& \kappa_{j}\star \overline{m_{i+1}}\cdot t_{i+1}\cdot \alpha_{j} 
& (\textrm{with } j\le i \textrm{ and } 
\varphi(j)\sqsubset \varphi(i+1)) \\ 
  &t_f\star \overline{n_f}\cdot\kappa_f\cdot\alpha_f & \kappa_s \star \alpha_s 
  & (\textrm{with } \M \vDash f(\vec{m}_{\varphi(s)},\vec{n}_{\varphi(s)})=0)
\end{array}$
\end{center}

We assume $t_0$ is not a winning strategy, that is there exists a term $u_0$ and a stack $\pi_0$ such that
$$\pair{t_0\star u_0\cdot \pi_0,\emptyset}\notin\W^{1}_{\Phi}$$
and try to reach a contradiction.

We build by induction four infinite sequences $(n_i)_{i\in\N},(u_i)_{i\in\N},(\pi_i)_{i\in\N},(\H_i)_{i\in\N}$ 
such that for any index $i\in\N$, we have $\H_i=\bigcup_{j\le i}\{(\vec{m}_{\varphi_i(j)},\vec{n}_{\varphi_i(j)},u_j,\pi_j)\}$
and the following statement:
$$\pair{t_i\{\kappa_j:=u_j,\alpha_j:=\pi_j\}_{j=0}^{i-1}\star \overline{n_i} \cdot u_i\cdot \pi_i,\H_i}\notin\W^{1}_{\Phi}
 \eqno{(\text{IH}_i)}$$

where $t_i$ is the term taken from the thread scheme we obtain for the sequence $(n_i)_{i\in\N}$.

\begin{itemize}
 \item {\bf IH$_1$} : by substitution over the first line of the scheme, we get 
$$t_{0} \star u_{0}\cdot\pi_{0} \eval u_{0}\star\overline{m}_1\cdot t_{1}\{\kappa_0:=u_0, \alpha_0:=\pi_0\}\cdot\pi_{0}$$
As $\pair{t_0\star u_0\cdot \pi_0,\emptyset}\notin\W^{1}_{\Phi}$, that implies by the second rule of induction that there exists $n_1,u_1,\pi_1$ such that 
$$\pair{t_{1}\{\kappa_0:=u_0, \alpha_0:=\pi_0\}\star \overline{n_1}\cdot u_1\cdot \pi_1,(\emptyset,\emptyset,u_0,\pi_0)}\notin\W^{1}_{\Phi}$$
\item {\bf IH$_{i+1}$} : assume we have built $n_j,u_j,\pi_j,\H_j$ for all $0\le j \le i$, such that IH$_j$ holds.
Hence by hypothesis, we have 
$$\pair{t_i\{\kappa_j:=u_j,\alpha_j:=\pi_j\}_{j=0}^{i-1}\star \overline{n_i} \cdot u_i\cdot \pi_i,\H_i}\notin\W^{1}_{\Phi}$$
By substitution over the threads scheme, we get an index  $j\le i$ such that :
\small
$$ t_i\{\kappa_j:=u_j,\alpha_j:=\pi_j\}_{j=0}^{i-1}\star \overline{n_1} \cdot u_i\cdot \pi_i
\eval 
u_j\star \overline{m_{i+1}}\cdot t_{i+1}\{\kappa_j:=u_j,\alpha_j:=\pi_j\}_{j=0}^{i} \cdot \pi_j$$
\normalsize
Furthermore we know from the hypothesis IH$_i$ that there is a pair $(\vec{m}_{\varphi_i(j)},\vec{n}_{\varphi_i(j)})$ 
such that $(\vec{m}_{\varphi_i(j)},\vec{n}_{\varphi_i(j)},u_j,\pi_j)\in\H_i$.
As the second rule of induction fails, it implies the existence of $n_j,u_j,\pi_j$ such that :
$$\pair{t_{i+1}\{\kappa_j:=u_j,\alpha_j:=\pi_j\}_{j=0}^{i}\star \overline{n_{i+1}} \cdot u_{i+1}\cdot \pi_{i+1},\H_{i+1}}\notin\W^{1}_{\Phi}$$
where, taking the very same definition of $\varphi_{i+1}$ we used in the proof of lemma \ref{lm:scheme},
$\H_{i+1}=\H_{i}\cup\{((\vec{m}_{\varphi_{i+1}(i+1)},\vec{n}_{\varphi_{i+1}(i+1)}),u_{i+1},\pi_{i+1})\}$, so we prove IH$_{i+1}$.
\end{itemize}

Now, if we consider the sequence $(n_i)_{i\in\N}$ we built, and define $\varphi=\lim_{i\in\N} \varphi_i$,
it is clear that $\varphi$ is the very same function that we obtain by Lemma \ref{lm:scheme}.
Moreover, according to this Lemma we know there exists $f,s\in\N$ such that 
$$t_f\{\kappa_j:=u_j,\alpha_j:=\pi_j\}_{j=0}^{f}\star \overline{n_f}\cdot u_f\cdot \pi_f \eval u_s \star \pi_s$$
with $\M \vDash f(\vec{m}_{\varphi(s)},\vec{n}_{\varphi(s)})=0$.
As $(\vec{m}_{\varphi(s)},\vec{n}_{\varphi(s)},u_s,\pi_s)\in\H_f$, the first rule of $\G^{1}_\Phi$ applies, and
$$\pair{t_f\{\kappa_j:=u_j,\alpha_j:=\pi_j\}_{j=0}^{f-1}\star \overline{n_f}\cdot u_f\cdot \pi_f,\H_f}\in\W^{1}_{\Phi}$$
which is obviously a contradiction with IH$_f$.
\end{proof}

\subsection{A wild realizer}
\label{ss:wild}

The previous section gives a specification of arithmetical formul\ae~in the particular case where the language
of realizers is deterministic\footnote{Actually, this assumption is not necessary, and has been made only for convenience in the proof of Lemma \ref{lm:scheme}.
In fact, we could adapt this proof to a non-deterministic case, by defining $Q_{i+1}$ as the union of the threads 
$\Th(t_i\star \overline{n_{i+1}}\cdot\kappa_{i+1}\cdot\alpha_{i+1})$ for all 
 $j\in\N, m_{i+1}\in\N$ and $t_{i+1}\in\Lambda$ such that $\kappa_j\star\overline{m_{i+1}}\cdot t_{i+1} \cdot \alpha_j \in Q_i$. 
 But in this case the characteristic function of the tree describing the thread scheme is more subtle to construct.}
and provides infinitely many interaction constants and infinitely many substitutive and non generative stack constants.
These assumptions are actually incompatible with the presence of instructions such as $\Eq$ or $\Quote$, 
as stated by the Proposition \ref{prop:ExtraInstructions}, since this break the property of substitutivity.  
It would be pleasing to be able to extend such a characterization to a more general framework that would allow such instructions.
Nevertheless, we know from \cite{Peirce} that it was not possible for the Law of Peirce, 
and it is not possible either in this case, for the very same reason:
the instruction $\Eq$ (that could be simulated with $\Quote$, see Section \ref{ss:ExtraInstr}) allows to
define some \emph{wild} realizers for some formul\ae, 
that is realizers of some $\Phi$ that are not winning strategies for the game $\G^{1}_{\Phi}$.

If we consider $f_{\le}:\N^{2}\to\N$ such that $\forall x,y\in\N, (f_{\le}(x,y)=0 \Leftrightarrow x\le y)$, 
and the formula ${\Phi_\le}\equiv \exr{x}\far{y}(f_{\le}(x,y)=0)$, here is an example of such a wild realizer.
We define the following terms
$$\begin{array}{rcl}
T_2[y,m] & \equiv & \Quote~(\lambda nu.\texttt{eq\_nat}~ n~ m~ (\Eq~ u~ (y~ y)~ \I~ u)~ u)\\
T_1[u,m] & \equiv & \lambda y.u~ \barr 0 ~ T_2[y,m]\\
T_0[u,m] & \equiv & T_1[u,m] ~T_1[u,m]\\
t_\le & \equiv & \lambda u.\Quote ~(\lambda m.T_0[u,m])
\end{array}$$
From these definitions we get for all $u\in\Lambda$ and $\pi\in\Pi$:
$$t_\le \star u \cdot \pi \eval T_0[u,\barr {n_\pi}]\star \pi 
\eval u \star \barr 0 \cdot T_2[T_1[u,\barr {n_{\pi}}],\barr {n_{\pi}}]\cdot \pi$$
and moreover, for all $n\in \N, u'\in\Lambda$ and $\pi'\in\Pi$:
$$T_2[T_1[u,\barr {n_{\pi}}],\barr {n_{\pi}}]\star \barr n \cdot u' \cdot \pi' \eval 
\left\{\begin{array}{ll} 
  \I \star \pi' & \mathrm{if~} u'\equiv T_0[u,\barr {n_\pi}] \mathrm{~and~} \pi\equiv\pi' \\
  u' \star \pi' & \mathrm{otherwise}
\end{array}\right.$$ 

\begin{prop}\label{prop:wild}
 $t_\le\ureal \exr{x}\far{y}(f_{\le}(x,y)=0)$
\end{prop}
\begin{proof}
 Let us consider a fixed pole $\pole$ and a stack $u\cdot \pi \in\fv{\exr{x}\far{y}(f_{\le}(x,y)=0)}$,
 that is a falsity value $S$ such that $\pi\in\fv{\dot S}$ and $u\in\tv{\far{x}(\far{y}(f_{\le}(x,y)=0)\limp \dot S)}$.
 We distinguish two cases:
 \begin{itemize}
  \item either $T_0[u,\barr {n_\pi}]\star \pi\in\pole $.
  As we have $t_\le\star u \cdot \pi \eval T_0[u,\barr {n_\pi}]\star \pi$, we get $t_\le\star u \cdot \pi \in\pole$ by anti-evaluation.
  \item either $T_0[u,\barr {n_\pi}]\star \pi\notin\pole $.
  In this case, we have $t_\le \star u \cdot \pi\eval u \star \barr 0 \cdot T_2[T_1[u,\barr {n_{\pi}}],\barr {n_{\pi}}]\cdot \pi$,
  hence it suffices to prove that $T_2[T_1[u,\barr {n_{\pi}}],\barr {n_{\pi}}]\real \far{y}(f_{\le}(0,y)=0)$.
  Let us then consider $n\in \N$ and a stack $u'\cdot\pi'\in\fv{\forall W (W (f_{\le}(0,n)) \limp W (0))}$.
  First remark that $f_{\le}(0,n)=0$, hence by Corollary \ref{cor:equality} $u'\star\pi'\in\pole$,
  thus by assumption, we know that $(u',\pi')\nequiv (T_0[u,\barr {n_\pi}],\pi)$.
  Thus we have $T_2[T_1[u,\barr {n_{\pi}}],\barr {n_{\pi}}]\star \barr n \cdot u' \cdot \pi' \eval u'\star\pi'\in\pole$, 
  which allows to conclude by anti-evaluation. \qedhere
 \end{itemize}
 \end{proof}
 
 Notice that the subterm $\I$ that appears in the definition of the term $T_2$ never comes to active position 
 in the proof of Proposition \ref{prop:wild}, so that we could actually have chosen any other closed $\lambda_c$-term instead.
 The point is that it can only occur if $(u',\pi')\equiv (T_0[u,\barr {n_\pi}],\pi)$, 
 and when it is the case, we are no more interested in the end of the execution of the process $T_0[u,\barr {n_\pi}]\star\pi$,
 that is in a way allowed to do anything in the rest of its execution.
 Before giving a game-theoretic interpretation of this phenomena, we first check that $t_\le$ is not a winning strategy for the game $\G^{1}_{\Phi_\le}$.
 
 \begin{prop}
 Let us assume that the relation of one step evaluation $\eval_1$ is only defined from the rules (\textsc{Grab}),
 (\textsc{Push}),(\textsc{Save}),(\textsc{Restore}),(\textsc{Quote}),(\textsc{Eq}). 
 Then the universal realizer $t_\le$ of $\Phi_\le$ is not a winning strategy for the game $\G^{1}_{\Phi_\le}$
 \end{prop}
 \begin{proof}
  The following is a valid match for $\G^{1}_{\Phi_\le}$ that \elo\ loses :
  \begin{itemize}
   \item \abe\ starts with the initial handle $(\I,\alpha)$ for the empty position, where $\alpha$ is a stack constant. 
   \item The only pair $(m,t)$ such that $t_\le \star \I\cdot\alpha \eval \I \star \barr {m} \cdot t \cdot \alpha$ is 
  $t_1\equiv T_2[T_1[\I,\barr {n_\alpha}],\barr {n_\alpha}]$ and $m_1=0$. Thus \elo\ is forced to play that pair $(0,t_1)$
   \item \abe\ replies with $n_1=0$, $u_1\equiv T_0[\I,\barr {n_\alpha}]$ and $\pi_1\equiv \alpha$.
   \item Then \elo\ loses, as the thread $\Th(t_1\star \barr {n_1} \cdot u_1 \cdot \pi_1)$ contains no process of the form
   $\I\star \barr{m} \cdot t\cdot \alpha$ (to continue to play) or of the form $u_1\star\pi_1$ (to win the game).\qedhere
   \end{itemize}
 \end{proof}

\section{Non-substitutive case}
\label{s:nsubst}
\subsection{$\G^{2}_\Phi$: cumulative game}
Despite the wild realizer $t_\le$ of the formula $\Phi_\le$ is not a winning strategy for the corresponding game $\G^{1}_{\Phi_\le}$,
we can still think its computational behaviour in game-theoretic terms as follows.
If we observe closely what happens in the match we described in the proof of the previous Proposition, 
if \abe\ starts with $(u,\pi)$, to which \elo\ answers $(0,T_2[T_1[u,\barr {n_\pi}],\barr {n_\pi}])$,
\elo\ then does somehow the distinction between two cases over the next \abe\ answer $(n_1,u_1,\pi_1)$.
\begin{itemize}
 \item if $(u_1,\pi_1)\nequiv(T_0[u,\barr {n_\pi}],\pi)$, \elo\ simply pursues the execution to reach $u_1\star\pi_1$,
 which is a final winning position, as $0\le n_1$.
 \item if $(u_1,\pi_1)\equiv(T_0[u,\barr {n_\pi}],\pi)$, as no interesting move can be obtained from the current position, 
 \elo\ \emph{backtracks} to the former $\exists$-position $t_\le \star u \cdot \pi$, and now wins  since 
 $$t_\le \star u \cdot \pi \eval T_0[u,\barr {n_\pi}]\star \pi \equiv u_1\star \pi_1$$
 \end{itemize}
 That is to say that the term $t_\le$ can still be seen as a winning strategy if we give the right to \elo\ 
 to compute its move from any former $\exists$-position. This gives us a new game $\G^{2}_{\Phi}$, 
 in which \elo\ keeps track of all the previous $\exists$-positions encountered during the game.
 
 We thus define a $\G^{2}_{\Phi}$-state as a pair $\pair{P,\H}$, where $P$ is now a finite set of processes (intuitively, all $\exists$-positions, 
 including the current one), and $\H$ is exactly as in $\G^{1}_{\Phi}$.
 The set $\W^{2}_\Phi$ of winning positions is inductively defined as follows:

\begin{itemize}
\item if there is $p\in P$ and $(\vc{m}{h},\vc{n}{h},u,\pi)\in\H$ 
such that  $p\eval u\star\pi$ and $\M \vDash f(\vc{m}{h},\vc{n}{h})=0$
$$\infer[{\mbox{\scriptsize (Win)}}]{\pair{P,\H} \in \W^{2}_\Phi}{}$$
\vspace{-0.3cm}

\item if there is $p\in P$, $i<h$, $(\vc{m}{i},\vc{n}{i},u,\pi)\in\H$  and  $m'\in\N$ such that  $p\eval u\star \overline{m'} \cdot t \cdot \pi$:
 $$\infer[{\mbox{\scriptsize (Play)}}]{\pair{P,\H} \in \W^{2}_{\Phi}}
{\pair{P\cup\{t\star \overline{n'} \cdot u' \cdot \pi'\}, \H\cup\{(\vc{m}{i}\cdot m',\vc{n}{i}\cdot n',u',\pi')\}} \in \W^{2}_\Phi
& \forall (n',u',\pi')\in\N\times\Lambda\times\Pi}$$
\end{itemize}

A term $t$ is say to be a \emph{winning strategy} for $\G^{2}_\Phi$ if for any handle $(u,\pi)\in\Lambda\times\Pi$, 
we have $\pair{\{t\star u\cdot\pi\}, \{(\emptyset,u,\pi)\}}\in\W^{2}_\Phi$.

\subsection{Adequacy}

\begin{prop}\label{prop:inclusion}
 A winning strategy for $\mathds{G}_\Phi^{1}$ is also a winning strategy for $\mathds{G}_\Phi^{2}$.
\end{prop}
\begin{proof}
 It suffices to prove that for any $\G^{1}_\Phi$ state $\pair{p,\H}$, 
 if we have $\pair{p,\H}\in\W^{1}_{\Phi}$, then $\pair{\{p\},\H}\in\W^{2}_\Phi$.
 We do it by induction on the derivation of $\pair{p,\H}\in\W^{1}_{\Phi}$, observing for the second rules of $\G^{2}_\Phi$ that
 if $\pair{P,\H}\in\W^{2}_{\Phi}$ and $P\subset P'$, then $\pair{P',\H}\in\W^{2}_{\Phi}$ (which is also proved by induction).
\end{proof}

\begin{prop}[Adequacy]\label{prop:adequacy}
 If $t$ is a winning strategy for $\mathds{G}^{2}_\Phi$, then $t\ureal \Phi$
\end{prop}
\begin{proof}
To make the proof easier, we will use the formul\ae~$A$ and $E$ that we previously defined in Section \ref{ss:g1}.

 Let $\pole$ be a fixed pole, $S\!_1$ be a falsity value, $u_0\real\far{x_1} (E_1\Rightarrow \dot{S}\!_1)) \Rightarrow \dot{S}\!_1$ and $\pi_0\in S\!_1$,
 and let us show that $t\star u_0\cdot\pi_0\in\pole$.
 For that, we more generally prove the following statement:
 \begin{clm}
 If $\pair{P,\H}\in\W^{2}_\Phi$ and  $\forall (\vc{m}{i},\vc{n}{i},u_i,\pi_i)\in\H, u_i\cdot\pi_i\in\fv{E_{i}\{x_j:=m_j,y_j:=n_j\}^{i}_{j=1}}$ 
 then  $P\cap\pole\neq\emptyset$  
 \end{clm}
\begin{proof}
We proceed by induction on the derivation of $\pair{P,\H}\in\W_\Phi$, distinguishing two possible cases:
\begin{enumerate}
 \item $\pair{P,\H}\in\W^{2}_\Phi$ because of the first induction rule:
  there exists $(\vc{m}{h},\vc{n}{h},u,\pi)\in\H$ and $p\in P$ such that
 $p\eval u\star\pi$ and $\M \vDash f(\vc{m}{h},\vc{n}{h})=0$. 
 If we assume that $u\cdot\pi\in\fv{E_h}=\fv{\forall W(W(f(\vc{m}{h},\vc{n}{h}))\Rightarrow W(0))}$,
 as $\M \vDash f(\vc{m}{h},\vc{n}{h})=0$, we get that $u\star\pi\in\pole$ (Corollary \ref{cor:equality}) and by anti-reduction, $p\in\pole$.
 
 \item $\pair{P,\H}\in\W^{2}_\Phi$ because of the second induction rule :
 there is some $p_i\in P$, $(\vc{m}{i},\vc{n}{i},u_i,\pi_i)\in\H$ and $m\in\N$ 
 such that  $p_i\eval u_i\star \overline{m} \cdot \xi \cdot \pi_i$, and for any $(n,u,\pi)$,
 $\pair{P\cup\{\xi\star \overline{n} \cdot u \cdot \pi\}, \H\cup\{(\vc{m}{h},\vc{n}{h},u,\pi)\}} \in \W^{2}_\Phi$.
 We prove that we can not have $P\cap\pole=\emptyset$.
 Indeed, assuming it is the case, we can show that $u_i\star \overline{m} \cdot \xi \cdot \pi_i\in\pole$.
 Besides, we know by hypothesis that 
 \begin{center}
  $u_i\cdot\pi_i\in\fv{\forall X_{i+1} (\far{x_{i+1}} 
	(\far{y_{i+1}}E_{i+1}\{{x}_{j}:={m}_{j},{y}_{j}:={n}_{j}\}^{i}_{j=1} \Rightarrow X_{i+1}) \Rightarrow X_{i+1})}$
  \end{center}
 so that it is sufficient to prove that $\xi\real\far{y_{i+1}}E_{i+1}\{{x}_{j}:={m}_{j},{y}_{j}:={n}_{j}\}^{i}_{j=1}\{x_{i+1}:=m\}$ to conclude.
 So pick $n\in\N$, $u\cdot\pi\in\fv{E_{i+1}\{{x}_{j}:={m}_{j},{y}_{j}:={n}_{j}\}^{i}_{j=1}\{x_{i+1}:=m\}\{y_{i+1}:=n\}}$, 
 and let us prove that $\xi\star\overline{n}\cdot u\cdot\pi\in\pole$.
 We have by hypothesis that 
 $$\pair{P\cup\{\xi\star \overline{n} \cdot u \cdot \pi\}, \H\cup\{(\vc{m}{i}\cdot m,\vc{n}{i}\cdot n,u,\pi)\}} \in \W^{2}_\Phi$$
 from which we deduce by induction (the premises are verified) that 
 $$(P\cup\{\xi\star \overline{n} \cdot u \cdot \pi\})\cap\pole\neq\emptyset$$
 As $P\cap\pole=\emptyset$, we get that $\xi\star \overline{n} \cdot u \cdot \pi\in\pole$, which conclude this case. \qedhere
\end{enumerate}
\end{proof}
 
 In particular, we have $\pair{\{t\star u_0\cdot\pi_0\},\{(\emptyset,\emptyset,u_0,\pi_0)\}}\in\W^{2}_\Phi$, $u_0\cdot\pi_0\in\fv{E_0}$, 
 hence we can deduce that $t\star u_0\cdot\pi_0\in\pole$.
 
\end{proof}

\subsection{Completeness of $\mathds{G}^{1}_\Phi$}
\begin{prop}[Completeness of $\mathds{G}^{2}_\Phi$]\label{prop:completeness2}
 If $t\ureal\Phi$ then $t$ is a winning strategy.
\end{prop}
\begin{proof}
 Let us reason by contradiction by assuming that there exists a handle $(u_0,\pi_0)\in\Lambda\times\Pi$
such that $\pair{t\star u_0\cdot\pi_0, \{(\emptyset,\emptyset,u_0,\pi_0)\}}\notin\W^{2}_\Phi$.
We will construct an increasing sequence $(\pair{P_j,\H_j})_{j\in\N}$ such that for any $j\in\N$, $\pair{P_j,\H_j}\notin \W^{2}_{\Phi}$.
For that, let us pick a fixed enumeration $\phi : \N \to \N\times\Lambda$ such that every pair $(m,t)$ appears infinitely many times in the range of $\phi$.
The sequence $(\pair{P_j,\H_j})$ is then defined as follows:
\begin{itemize}
 \item We set $P_0=\{t\star u_0\cdot\pi_0\}$ and $\H_0=\{(\emptyset,\emptyset,u_0,\pi_0)\}$.
 \item Assume we have built a state $\pair{P_j,\H_j}\notin\W^{2}_\Phi$. Writing $(m,t)=\phi(j)$, we distinguish the two following cases:
 \begin{enumerate}
  \item Either there exists $p\in P_j$ and $((\vc{m}{i},\vc{n}{i},u,\pi)\in\H_j)$ such that  $p\eval u\star \overline{m} \cdot t \cdot \pi$. 
  From the second rule of induction we get the existence of $n\in\N$, $u'\in\Lambda$, $\pi'\in\Pi$ such that 
  $\pair{P\cup\{t\star \overline{n} \cdot u' \cdot \pi'\}, \H\cup\{(\vc{m}{i}\cdot m,\vc{n}{i}\cdot n,u',\pi')\}} \notin \W^{2}_\Phi$.
  We pick such a tuple $(n,u',\pi')$ and define $P_{j+1}=P_j\cup\{t\star \overline{n} \cdot u' \cdot \pi'\}$ 
  and $\H_{j+1}=\H_j\cup\{(\vc{m}{i}\cdot m,\vc{n}{i}\cdot n,u',\pi')\}$. 
  \item Either there is no such process, and we set $P_{j+1}=P_j$ and $\H^{j+1}=\H_j$. 
 \end{enumerate}
 In both cases, we have construct $P_{j+1}$ and $\H_{j+1}$ such that
 $P_j\subset P_{j+1}$, $\H_j\subset\H_{j+1}$ and $\pair{P_{j+1},\H_{j+1}}\notin \W^{2}_{\Phi}$.
 We set $P_\infty=\union{j\in\N}P_j$, $Q=\union{p\in P_\infty} \Th(p)$ and $\pole=Q^{c}$.
\end{itemize}
 
By construction, we have $t\star u_0\cdot\pi_0\notin\pole$, and as $t\ureal\forall X (\far{x_1}(\far{y_1}E_1\Rightarrow X )\Rightarrow X)$, 
we get  $u_0\nVdash\far{x_1}(\far{y_1}E_1\Rightarrow\{\pi_0\})$. 
Thus there exists $m_1\in\N$ and  $\xi_1\real \far{y_1}E_1\{x_1:=m_1\}$ such that $u_0\star \overline{m_1} \cdot \xi_1\cdot\pi_0\notin  \pole$,
that is exists an index $j\in\N$ and a process $ p\in P_j$ such that $p\eval u_0\star \overline{m_1} \cdot \xi_1\cdot\pi_0$. 
Let $k\ge j$ be such that $\phi(k)=(m_1,\xi_1)$, then by construction there is some $\overline{n_1}, u_1, \pi_1$
such that $P_{k+1}=P_k\cup\{\xi_1\star\overline{n_1}\cdot u_1 \cdot \pi_1\}$ and $\H_{k+1}=\H_k\cup\{((m_1,n_1,u_1,\pi_1)\}$

As $\xi_1\real \far{y_1}E_1\{x_1:=m_1\}\equiv \far{y_1}\forall X ((\far{x_2}\far{y_2}E_2\{x_1:=m_1\}\Rightarrow X )\Rightarrow X)$ 
and $\xi_1\star\overline{n_1}\cdot u_1 \cdot \pi_1\notin\pole$, we deduce than $u_1\nVdash\far{x_2}\far{y_2}E_2\{x_1:=m_1,y_1:=n_1\}\Rightarrow\{\pi_1\})$.

Iterating this very same reasoning, we obtain that for every $i\in\llbracket 1,h\rrbracket$,
there exists an index $k_i\in\N$ and a closed term $\xi_i\in\Lambda$, such that $\H_{k_i}$ contains a tuple
$(\vc{m}{i},\vc{n}{i},u_i,\pi_i)$, with $\xi_i\star \overline{n_i}\cdot u_i \cdot \pi_i \notin\pole$
and $\xi_i\real \far{y_i}E_{i}\{x_j:=m_j\}^{i}_{j=1}\{y_j:=n_j\}^{i-1}_{j=1}$.

For $i=h$, we get then an index $k_h\in\N$ and a closed term $\xi_h$, such that $\H_{k_h}$ contains a tuple
$(\vc{m}{h},\vc{n}{h},u_h,\pi_h)$, with $\xi_h\star \overline{n_h}\cdot u_h \cdot \pi_h \notin\pole$
and $\xi_h\real \far{y_h}\forall W (W(f(\vc{m}{h},\vc{n}{h-1}\cdot y_h))\Rightarrow W(0))$.

If we consider the following predicate 
$$\Delta:\left\{\begin{array}{ccl}\N & \to & \Pow(\Pi) \\ 0 & \mapsto & \{\pi_h\} \\ n\ge1 & \mapsto& \emptyset\end{array}\right.$$
we get in particular that 
$\xi_h\real \{n_h\}\limp \Delta(f(\vc{m}{h},\vc{n}{h}))\Rightarrow \Delta(0) $, 
from which we deduce that $u_h \cdot \pi_h \notin \fv{\Delta(f(\vc{m}{h},\vc{n}{h}))\Rightarrow \Delta(0)}$.
Obviously $\pi_h\in\fv{\Delta(0)}$, so that necessarily we have $u_h\nVdash\Delta(f(\vc{m}{h},\vc{n}{h}))$.
Hence there exists $\pi\in\fv{\Delta(f(\vc{m}{h},\vc{n}{h}))}$, which implies that $\pi=\pi_h$ and $\M\vDash f(\vc{m}{h},\vc{n}{h})=0$,
such that $u_h\star\pi_h\notin \pole$, that is to say there is some $j\in\N$ and $p\in P_j$ such that $p\eval u_h\star \pi_h$. 
Taking $l=\max(j,k_h)$, this contradicts the fact that $(P_l,\H_l)\notin \W^{2}_\Phi$ because of the first rule of induction.
\end{proof}

\begin{thm}
If $\Phi$ is an arithmetical formula, there exists $t\ureal \Phi$ if and only if $t$ implements a winning strategy for $\G^{2}_\Phi$.
\end{thm}

\section{A barrier for realizability models}
\label{s:models}
\newcommand{\liste}[2]{\<#1\>_{#2}}
\newcommand{\uncode}{\check}
\newcommand{\code}{\hat}
\subsection{A universal realizer for every formul\ae}
\label{ss:Universal}
We show here that if an arithmetic formula $\Phi\equiv\exr{x_1}\ldots\far{y_h}f(\vc{m}{h},\vc{n}{h})=0$ is true in the ground model, 
as soon as we dispose of a term computing $f$, we can implement a winning strategy, hence a universal realizer.
The idea of the strategy for \elo\  is to enumerate "smartly" $\N^{h}$, in the following sense:
when playing a tuple $\vc{m}{h}$, we first look as deep as possible in the tree of formers positions 
for the tuple $\vc{m}{i}$, and then go with corresponding \abe\ answer. 
In doing so we ensure that any tuple $\vc{m}{i}$ will always be played with the same answers $\vc{n}{i}$.
Then it is clear that is $\M\vDash\Phi$, we will reach sooner or later a winning position.

To implement such a strategy, we consider a term computing $f$ on a given position :
$$\Theta_f\star \liste{\barr m}{h} \cdot t_1 \cdot t_2 \cdot \pi \eval
\begin{cases}
  t_1\star\pi & \textrm{if } \M\vDash f(\vec{m},\vec{n})=0 \\
  t_2\star\pi & \textrm{if } \M\vDash f(\vec{m},\vec{n})\neq0 \end{cases}$$
where $\liste{\barr m}{i}$ is a $\lambda_c$-implementation\footnote{We could chose for instance to use a list representation for tuples, 
in which case $\liste{\barr m}{i}\equiv [\barr m_1,\ldots,\barr m_i]$,
but here the data-type would not be relevant, we only pay attention to some "big" steps of reduction independently of
technical representation of data}
for the tuple $\vc{m}{i}$,
and that we also have a term $\Next$ acting as a successor for $\N^{h}$. 
$$\Next \star \liste{\barr {m}}{h}^{i}\cdot t \cdot \pi \eval t \star \liste{\barr {m}}{h}^{i+1}\cdot\pi$$
where $\vc{m}{h}^{0}=(0,\ldots,0)$ and the sequence $(\vc{m}{h}^{i})_{i\in\N}$ 
is an enumeration of $\N^{h}$.
We also define the relation $\vc{m}{h}^{i}\le_h\liste{m}{h}^{j}\equiv i\le j$, which is total on $\N^{h}$.
Furthermore, we assume that we dispose of a $\lambda_c$-implementation of histories as lists of tuples,
and for a given history $H$, we will denote by $\code H$ its implementation\footnote{
$H\cup\{(\vc{m}{i},\vc{n}{i},u,\pi)\}$ will so correspond to $[\liste{\barr m}{i},\liste{\barr n}{i}),u,k_\pi]\cdot \code H$\,.}.

\begin{definition}
We say that a history $H$ is  \emph{functional} if for any $\vc{m}{i}$,
there exists at most one tuple $(\vc{n}{i},u,\pi)$ such that $(\vc{m}{i},\vc{n}{i},u,\k_\pi)\in H$.
\end{definition}

Then we build\footnote{We let the reader check the existence of such terms, which is a straightforward $\lambda_c$-calculus exercise}
several $\lambda_c$-terms according to their reductions rules. 
These terms will all take as parameter a $\lambda_c$-history $\code H$.
For $1\le i < h$, we define a term $T_i$ who is intended to
gets \abe` $i^{\text{th}}$ answer $(n_i,u_i,\pi_i)$, save it in $\code H$ and plays the next integer with $T_{i+1}$:
$$\begin{array}{r@{~~\eval~~}l}
T_i[\vc{m}{h},\vc{n}{i-1},\code H]\star \barr n_i \cdot u_i \cdot \pi_1 
      & u_i \star \barr m_{i+1}\cdot T_{i+1}[\vc{m}{h},\vc{n}{i},\code H^{(i)}] \cdot \pi_i
      \end{array}$$
where $\code H^{i}\equiv[\vc{m}{1},\vc{n}{i},u_i,\k_{\pi_i}]\cdot \code H$. 
The term $T_h$ gets \abe's answer as $T_i$ does, but then computes $f$ to know if it has reached a winning position 
or should either initiate the next step of enumeration:
$$T_{h}[\vc{m}{h},\vc{n}{h-1},\code H] \star \barr n_h \cdot u_h \cdot \pi_h 
      ~\eval~ \Theta_f \star \liste{m}{h}\cdot\liste{n}{h}\cdot u_h \cdot N[\vc{m}{h},\code H^{(h)}]\cdot\pi_h$$
with $\code H^{h}\equiv[\vc{m}{h},\vc{n}{h},u_h,\k_{\pi_h}]\cdot \code H$.
Then $N$ computes the next tuple in the enumeration and $L$ looks in the tree for the maximum former partial position similar 
to an initial segment of this tuple:
$$\begin{array}{r@{~~\eval~~}l}
N[\liste{m}{h},\code H]\star\pi
      & \Next \star \liste{m}{h} \cdot (\lambda m'_1\cdots m'_h.L[\liste{m'}{h},\code H])\cdot\pi \\
      L[\liste{m}{h},\code H]\star\pi
      & u_i \star \barr m_{i+1}\cdot T_{i+1}[\liste{m}{h},\liste{n}{i},\code H] \cdot \pi_i
      \end{array}$$
with $(\liste{m}{i},\liste{n}{i},u_i,\k_{\pi_i})\in \code H$ and 
$\forall j>i, \forall \vc{n}{j}\in \N^{j},\forall u\in\Lambda,\forall \pi\in\Pi (\liste{m}{j},\liste{n}{j},u,\k_{\pi})\notin \code H$.
Finally we consider $t_\Phi$ that would be the winning strategy, such that:
$$t_\Phi\star u_0\cdot\pi_0 
      ~\eval~ u_0\star \barr 0\cdot  T_1[\liste{0}{h},\emptystr,\code H_0] \cdot \pi_0$$
   with $H_0\equiv (\cdot,\cdot,u_0,\k_{\pi_0})$

\begin{prop}\label{prop:RealUniversal}
 If $\M\vDash\Phi$, then $t_\Phi$ is a winning strategy for $\G^{1}_\Phi$.
\end{prop}

The proof does neither present any conceptual difficulty nor any interest in itself,
but still remains quite technical.
The idea is to propagate the contradiction along the enumeration of $\N^{h}$ 
in order to contradict $\M\vDash\Phi$ at the limit.
To do so, we define the proposition
{\bf P}$(i,\vc{m}{h},H)$ as the following statement :\\
{\bf P}$(i,\vc{m}{h},H)$ :"there exists $\vc{n}{i}\in\N^{i}$, $u_i\in\Lambda$, $\pi_i\in\Pi$ such that
\begin{itemize}
 \item $\{(\vc{m}{i},\vc{n}{i},u_i,\pi_i)\}\cup H$ is functional 
 \item $\pair{T_i[\liste{m}{h},\liste{n}{i-1},\code H]\star \barr n_i\cdot u_i\cdot\pi_i, 
                     \{(\vc{m}{i},\vc{n}{i},u_i,\pi_i)\}\cup H}\notin\W^{1}_{\Phi}$"
\end{itemize}

and prove two technical lemmas.   

\begin{lem}\label{lm:tech1}
 For any $i\in\Int{1,h}$, $\vc{m}{h}\in\N^{h}$ and any history $H$, {\bf P}$(i,\vc{m}{h},H)$ implies
 there exists an history $H'$ such that $H\subset H'$ and {\bf P}$(h,\vc{m}{h},H')$
\end{lem}
\begin{proof}
 It suffices to see that because of the reduction rule defining $T_i$, 
 if {\bf P}$(i,\vc{m}{h},H)$ holds then the second rule of $\G^{1}_\Phi$ has to fail, 
hence there exists $n_{i+1}$, $u_{i+1}\in\Lambda$, $\pi_{i+1}\in\Pi$ such that
\begin{equation*}
 \pair{T_{i+1}[\liste{m}{h},\liste{n}{i},\code H^{i}]\star \barr n_{i+1}\cdot u_{i+1}\cdot\pi_{i+1},
 \{(\vc{m}{i+1},\vc{n}{i+1},u_{i+1},\pi_{i+1})\}\cup H^{i}}\notin\W^{1}_{\Phi}
\end{equation*}
 where $H^{i}\equiv \{\vc{m}{1},\vc{n}{i},u_i,\k_{\pi_i}\}\cup  H$, which is still a functional environment.
 Therefore {\bf P}$(i,\vc{m}{h},H)$  $\limp$  {\bf P}$(i+1,\vc{m}{h},H^{i})$, and
 {\bf P}$(i,\vc{m}{h},H)$  $\limp$  {\bf P}$(h,\vc{m}{h},H')$) follows by easy decreasing induction on $i\in\Int{1,h}$.
\end{proof}

\begin{lem}\label{lm:tech2}
For any history $H$, {\bf P}$(h,\vc{m}{h}^{j},H)$ implies that 
\begin{enumerate}
  \item there exists $\vc{n}{h}\in\N^{h}$ such that$\M\nvDash f(\vc{m}{h}^{j},\vc{n}{h})=0$
\item there exists a history $H'$ such that $H\subset H'$ and {\bf P}$(h,\vc{m}{h}^{j+1},H')$
 \end{enumerate}
 \end{lem}
\begin{proof}
 Given a history $H$, if {\bf P}$(h,\vc{m}{h}^{j},H)$ holds, then it means that the first rule of induction of $\G^{1}$ fails,
hence necessarily $\M\nvDash f(\vc{m}{h}^{j},\vc{n}{h})=0$ and by definition of $T_h$, 
using the notations
$H^{h}=\{(\vc{m}{h}^{j},\vc{n}{h},u_{h},\pi_{h})\}\cup H$ and $\liste{m'}{h}=\liste{m}{h}^{j+1}$,
we get that 
$$T_{h}[\liste{m}{h}^{j},\liste{n}{h-1},\code H] \star \barr n_h \cdot u_h \cdot \pi_h 
      ~\eval~ 
      u_i \star \barr m'_{i+1}\cdot T_{i+1}[\liste{m'}{h},\liste{n'}{i},\code {H^{h}}] \cdot \pi_i
      $$
with $(\liste{m'}{i},\liste{n'}{i},u_i,\k_{\pi_i})\in \code H$ and 
$\forall j>i, \forall (\vc{n}{j},u,\pi)\in \N^{j}\times\Lambda\times\Pi (\liste{m'}{j},\liste{n}{j},u,\k_{\pi})\notin \code H$.
Note that this condition ensures the functionality of $H^{h}\cup\{(\vc{m}{h}^{j+1},\vc{n'}{i},u_{i},\pi_{i})\}$.
From {\bf P}$(h,\vc{m}{h}^{j},H)$ once more, we get that the second rule of induction of $\G^{1}$ fails too,
and so that {\bf P}$(i+1,\vc{m}{h}^{j+1},H^{h})$.
Hence by Lemma \ref{lm:tech1} we get the existence of $H'$ such that $H\subset H^{h}\subset H'$ 
and {\bf P}$(h,\vc{m}{h}^{j+1},H')$ holds.
\end{proof}

\begin{proof}[Proof of Proposition \ref{prop:RealUniversal}] By contraposition.
We show that if $t_\Phi$ is not a winning strategy, 
then there exists a growing sequence of history $(H_j)_{j\in\N}$ 
such that for all $j\in\N$, 
${\bf P}(h,\vc{m}{h}^{j},H_j)$ holds.

Indeed, assume $t_\Phi$ is not a winning strategy,
that is to say there is $u_0\in\Lambda,\pi_0\in\Pi$ 
such that $\pair{t_\Phi\star u_0\cdot \pi_0,\emptyset}\notin\W^{1}_\Phi$
Then because of the reduction rule of $t_\Phi$,
it means that the second rule of $\G^{1}_\Phi$ fails, thus there exists $(n_{1},u_{1},\pi_1\in\N\times\Lambda\times\Pi$, such that
$$\pair{T_{1}[\liste{0}{h},\emptystr,\code H]\star \barr n_{1}\cdot u_{1}\cdot\pi_{1}, 
 \{(\vc{m}{1},\vc{n}{1},u_1,\pi_1)\}\cup H}\notin\W^{1}_{\Phi}$$
 with $H\equiv (\cdot,\cdot,u_0,{\pi_0})$, 
that is ${\bf P}(1,\vc{m}{h}^{0},H)$.
 Then by Lemma \ref{lm:tech1} we get that 
 there exists $H_0$ such that ${\bf P}(h,\vc{m}{h}^{0},H_0)$ holds,
 and the claim follows by easy induction.
 \newcommand{\bigH}{\mathcal{H}}
 Then we set $\bigH=\bigcup_{j\in\N} H_j$, that is functional (because each $H_j$ is, and $H_j\subset H_{j+1}$).
 
 Applying the first clause of Lemma \ref{lm:tech2}, 
 we get that for all $j\in\N$, there exists $\vc{n}{h}^{j}$ such that 
 $(\vc{m}{h}^{j},\vc{n}{h}^{j},u,\pi)\in\bigH$ for some $u\in\Lambda$ and $\pi\in\Pi$
 and $\M\nvDash f(\vc{m}{h}^{j},\liste{n}{h}^{j})=0$.

 Furthermore, as $\bigH$ is functional, it easily implies that:
 $$\forall m_1 \exists n_1\ldots \forall m_h\exists n_h (\M\nvDash f(\vc{m}{h},\vc{n}{h})=0)$$
 and thus we finally get $\M\nvDash\Phi$.
\end{proof}

Combining the results we obtained at this point, we get the following theorem:

\begin{thm}\label{thm:main}
If $\Phi$ is an arithmetical formula, then $\M\vDash\Phi$ if and only if there exists $t\ureal \Phi$.
\end{thm}
\begin{proof}
 The first direction is a consequence of Propositions \ref{prop:RealUniversal} and \ref{prop:adequacy}, 
 the reverse directly comes from Proposition \ref{p:Degenerated}.
\end{proof}

\subsection{Leibniz equality vs primitive non-equality}
Here we have chosen to consider formul\ae\ based on equalities, 
and we should wonder what happens if we use instead formul\ae\ based on disequalities:
$$\exists x_1\forall y_1\ldots\exists x_n\forall y_n f(\vc{x}{n},\vc{y}{n})\neq 0\,.$$
We know that both definitions are equivalent from a model-theoretic point of view.
Indeed, if we define the following function $h$:
$$h=\begin{cases} x\mapsto 1 & \text{if } x=0\\
                   x\mapsto 0 & \text{otherwise}
                  \end{cases}$$
then for all $\vec{x}\in\N^{n}$, $\M \vDash f(\vec{x})=0$ if and only if $\M\vDash (h\circ f)(\vec{x})\neq 0$.
In other words, formul\ae~ based on a non-equality have the same expressiveness, 
and we also might have chosen it as definition for the arithmetical formul\ae~(see Definition \ref{def:formulae}).

In classical realizability the disequality can be a given a simple semantic:
$$\fv{e_1\neq e_2}=\begin{cases}\fv{\top} & \text{if } \M\vDash e_1\neq e_2\\
                   \fv{\bot} & \text{otherwise}
                  \end{cases}$$
which is equivalent to the negation of equality. Indeed, one can easily check that we have
$\lambda x t.(t)x\Vdash e_1\neq e_2 \limp \neg (e_1=e_2)$ and $\lambda t. (t)I\Vdash \neg(e_1=e_2)\limp e_1\neq e_2$.

Yet using these definitions, the rules of the game would have slightly changed. 
Indeed, if we observe closely what happens at the last level of the game (with every variable already instantiated 
but the one of the last universal quantifier), that is a formula $\far{y}(f(y)\neq 0)$, if the formula is true in the model,
then the falsity value is empty, so that the opponent can not give any answer:
$$\fv{\forall{y}(f(y)\neq 0)}=\union{n\in\N}\fv{f(n)\neq 0}=\fv{\top}=\emptyset
\eqno{(\forall n\in\N, \M\vDash f(n)\neq 0)}$$
Hence \elo\ does not have to compute the formula $f$ to know whether she can win or not, 
she only has to wait for a potential answer of \abe, and keep on playing if she eventually gets one.

We shall bring the reader to notice two important facts.
Firstly, it is clear that as \elo\  has no need to compute $f$, she only needs to do somehow a ``blind" enumeration,
hence we can build the very same realizer we built in Proposition \ref{prop:RealUniversal}
without using a term computing $f$.
In fact, such a realizer would be suitable for any $f$, even not computable, that is :
\begin{prop}{\cite{Kri09}}\label{prop:RealUniversal2}
 For all $n\in\N$, there exists $t_n\in\Lambda_c$ such that for any $f:\N^{2n}\to\N$, 
 if $\M\vDash\exists x_1\forall y_1\ldots\exists x_n\forall y_n f(\vc{x}{n},\vc{y}{n})\neq 0$, 
 then $t_n\Vvdash \exr x_1\far y_1\ldots\exr x_n\far y_n f(\vc{x}{n},\vc{y}{n})\neq 0$.
\end{prop}

Secondlyd, such a result it obviously false if we use equality instead of non-equality.
Going back to the halting problem, if we consider one of the functions $f:\N^{2}\to\N$ such that 
\begin{center}
\begin{tabular}{ccc}
$f(m,n)=0$ & \hspace{1cm}~iff~  & $(n=0 \wedge \exr p(\Halt(m,p))) \vee (n\neq0\wedge\far p(\neg \Halt(m,p)))$
\end{tabular}
\end{center}
it is clear that $f$ is not computable and that $\M\vDash \forall y \exists x (f(y,x)=0)$ 
(that only says that a Turing machine stops or does not stop).
We know by Proposition \ref{prop:RealUniversal2} that there is a term $u\in\Lambda_c$
such that $u\Vvdash \far{y}\exr{x} (h\circ f)(y,x)\neq0$, 
but there is no term\footnote{Otherwise, using a witness extraction method for $\Sigma_0^1$-formul\ae~\cite{Miq10}, 
we would be able for all $m\in\N$ to compute $n_m\in\N$ such that $f(m,n_m)=0$, breaking the halting problem.}
$t$ such that $t\Vvdash\far{y}\exr{x} f(y,x)=0$,
and thus no term $t'$ such that $t'\Vvdash(\far{y}\exr{x} (h\circ f)(y,x)\neq0)\limp (\far{y}\exr{x} f(y,x)=0)$.
This phenomena is quite strange\footnote{In fact, it already appears when considering the formula $\forall x(x=0\liff h(x)\neq 0)$
that is not realized if not relativized to naturals.}, as both formul\ae~ were perfectly equivalent in the ground model.
As we explained, a game-theoretic interpretation of this fact is based on the idea on the idea that the use of a non-equality
leaves the computation to the opponent, and making so the game easier. 
However, in the author's opinion this does not furnish a satisfying enough explanation for the model-theoretic point of view,
and it might be interesting to deal with this phenomena more deeply.

\subsection{Connection with forcing}
\label{ss:Forcing}

In this paper, we only considered the \emph{standard realizability models} of PA2 (following the
terminology of~\cite{RAlgI}), that is: the realizability models parameterized on tuples of the form
$(\Lambda,\Pi,{\eval},\Bot)$, where $(\Lambda,\Pi,{\eval})$ is a particular instance of the
$\lambda_c$-calculus, and where $\Bot$ is a pole.
The strong separation between the calculus (on one side) and the pole (on the other side) is
essential to define the notion of universal realizability, which is at the heart of the
specification problem studied in this paper.

However, the definitions of classical realizability can be extended in many different
ways.
First, we may replace second-order arithmetic (PA2) by Zermelo-Fraenkel set theory (ZF),
using a model-theoretic construction  \cite{Kri00,Gimmel} that is reminiscent from the construction
of forcing models and of Boolean-valued models of ZF.
\emph{Mutatis mutandis}, all the results presented in this paper remain valid in the framework
of classical realizability models of ZF, provided we consider a representation of
arithmetic formul{\ae} in the language of set theory that preserves their computational
interpretation in the sense of PA2 (see~\cite{Gimmel}).

Second, we may replace the terms and stacks of the $\lambda_c$-calculus by the
$\A$-terms and $\A$-stacks of an arbitrary \emph{classical
realizability algebra} $\A$, as shown by Krivine~\cite{RAlgI,Gimmel}.
Intuitively, classical realizability algebras generalize $\lambda_c$-calculi (with poles)
the same way as partial combinatory algebras \cite{Combinatory} generalize the $\lambda$-calculus (or
Gödel codes for partial recursive functions) in the framework of intuitionistic
realizability.
This broad generalization of classical realizability---in a framework where terms
and stacks are not necessarily of a combinatorial nature---is essential, since it allows
us to make the connection between forcing and classical realizability explicit.
Indeed, any complete Boolean algebra can be presented as a classical realizability
algebra, so that all Boolean-valued models of ZF (or forcing models) can 
actually seen as particular cases of classical realizability models of ZF.
(In this setting, the combination of realizability and forcing presented in
\cite{RAlgI,Forcing} can be seen as a generalization of the method of iterated forcing.)

In the general framework of classical realizability algebras, the specification
problem studied in this paper does not make sense anymore (due to the loss of
the notion of universal realizability), but we can still use the $\lambda_c$-terms
presented in Section~\ref{ss:Universal} to show more generally that every arithmetic formula that
is true in the ground model is realized by a proof-like term.

\begin{thm}
Let $\M$ be a Tarski model of ZFC, $\A$ a classical realizability
algebra taken as a point of $\M$, and $\M^{\A}$ the classical
realizability model of ZF built from the ground model $\M$ and the
classical realizability algebra $\A$.
Then for every closed arithmetical formula $\phi$ (expressed in theo
language of ZF) such that $\M\models\phi$, there exists a proof-like
term $\theta\in\A$ such that $\theta\Vdash_{\A}\phi$.
\end{thm}

This shows that arithmetical formul\ae\ remain absolute in the framework
of classical realizability models of set theory, which generalizes a
well-known property of forcing models to classical realizability.
Actually, recent work of Krivine~\cite{Kri14} shows that this result
extends to the class of $\Sigma^1_2$- and $\Pi^1_2$-formul\ae\ as well.
By discovering the existence of an ultrafilter for the characteristic
Boolean algebra $\gimel2$ \cite{Gimmel} of the realizability model~$\M^{\A}$,
Krivine succeeded to construct (by quotient and extensional collapse)
a proper class $\M'\subseteq\M^{\A}$ that constitutes a transitive
model of~ZF elementarily equivalent to~$\M$, and that contains the
same ordinals as~$\M^{\A}$.
Hence the Levy-Schoenfield theorem [14, Theorem 25.20] applies to
$\M$, $\M'$ and~$\M^{\A}$, thus proving the absoluteness of
$\Sigma^1_2$-and $\Pi^1_2$-formul\ae.

\medskip
\begin{center}
\rule{150pt}{0.5pt}
\end{center}

\subsubsection*{Acknowledgements}
The authors wish to thank Alexandre \textsc{Miquel}, 
who provided valuable assistance to the writing of the connection between the specification results and forcing in Section \ref{ss:Forcing}.
\medskip
\bibliographystyle{amsplain}
\bibliography{biblio}

\providecommand{\bysame}{\leavevmode\hbox to3em{\hrulefill}\thinspace}
\providecommand{\MR}{\relax\ifhmode\unskip\space\fi MR }
\providecommand{\MRhref}[2]{%
  \href{http://www.ams.org/mathscinet-getitem?mr=#1}{#2}
}
\providecommand{\href}[2]{#2}
\begin{thebibliography}{10}

\bibitem{BB96}
F.~Barbanera and S.~Berardi, \emph{A symmetric lambda calculus for classical
  program extraction}, Inf. Comput. \textbf{125} (1996), no.~2, 103--117.

\bibitem{Bar84}
H.~Barendregt, \emph{The lambda calculus: Its syntax and semantics}, Studies in
  Logic and The Foundations of Mathematics, vol. 103, North-Holland, 1984.

\bibitem{Chu41}
A.~Church, \emph{The calculi of lambda-conversion}, Annals of Mathematical
  Studies, vol.~6, Princeton, 1941.

\bibitem{Coquand95}
Thierry Coquand, \emph{A semantics of evidence for classical arithmetic}, J.
  Symb. Log. \textbf{60} (1995), no.~1, 325--337.

\bibitem{CH00}
P.-L. Curien and H.~Herbelin, \emph{The duality of computation}, ICFP, 2000,
  pp.~233--243.

\bibitem{CF58}
H.~B. Curry and R.~Feys, \emph{Combinatory logic}, vol.~1, North-Holland, 1958.

\bibitem{Fri73}
H.~Friedman, \emph{Some applications of {K}leene's methods for intuitionistic
  systems}, Cambridge Summer School in Mathematical Logic, Springer Lecture
  Notes in Mathematics, vol. 337, Springer-Verlag, 1973, pp.~113--170.

\bibitem{Fri78}
\bysame, \emph{Classically and intuitionistically provably recursive
  functions}, Higher Set Theory \textbf{669} (1978), 21--28.

\bibitem{Gir06}
J.-Y. Girard, \emph{Le point aveugle -- cours de logique -- volume {I} -- vers
  la perfection}, Hermann, 2006.

\bibitem{Gir89}
J.-Y. Girard, Y.~Lafont, and P.~Taylor, \emph{Proofs and types}, Cambridge
  University Press, 1989.

\bibitem{GuiPhD}
M.~Guillermo, \emph{Jeux de r{\'e}alisabilit{\'e} en arithm{\'e}tique
  classique}, Ph.D. thesis, Universit{\'e} Paris 7, 2008.

\bibitem{Peirce}
M.~Guillermo and A.~Miquel, \emph{Specifying peirce's law in classical
  realizability}, Submitted, 2011.

\bibitem{How69}
W.~A. Howard, \emph{The formulae-as-types notion of construction}, Privately
  circulated notes, 1969.

\bibitem{Combinatory}
P.~J. W.~Hofstra J.~R. B.~Cockett, \emph{An introduction to partial lambda
  algebras}, 2006.

\bibitem{Kle45}
S.~C. Kleene, \emph{On the interpretation of intuitionistic number theory},
  Journal of Symbolic Logic \textbf{10} (1945), 109--124.

\bibitem{Kri93}
J.-L. Krivine, \emph{Lambda-calculus, types and models}, Masson, 1993.

\bibitem{Kri00}
\bysame, \emph{The curry-howard correspondence in set theory}, LICS, IEEE
  Computer Society, 2000, pp.~307--308.

\bibitem{Kri01}
\bysame, \emph{Typed lambda-calculus in classical {Z}ermelo-{F}raenkel set
  theory}, Arch. Math. Log. \textbf{40(3)} (2001), 189--205.

\bibitem{Kri03}
\bysame, \emph{Dependent choice, `quote' and the clock}, Th. Comp. Sc.
  \textbf{308} (2003), 259--276.

\bibitem{Kam}
\bysame, \emph{A call-by-name lambda-calculus machine}, Higher Order and
  Symbolic Computation, 2004.

\bibitem{Kri09}
\bysame, \emph{Realizability in classical logic. {I}n interactive models of
  computation and program behaviour}, Panoramas et synth{\`e}ses \textbf{27}
  (2009).

\bibitem{RAlgI}
\bysame, \emph{Realizability algebras: a program to well order r}, Logical
  Methods in Computer Science \textbf{7} (2011), no.~3.

\bibitem{Gimmel}
\bysame, \emph{{Realizability algebras II : new models of ZF + DC}}, Logical
  Methods in Computer Science \textbf{8} (2012), no.~1, 10, 28 p.

\bibitem{Kri14}
J.-L. Krivine, \emph{{Quelques propri{\'e}t{\'e}s des mod{\`e}les de
  r{\'e}alisabilit{\'e} de ZF}}, February 2014,
  {http://hal.archives-ouvertes.fr/hal-00940254}.

\bibitem{McCPhD}
D.~McCarty, \emph{Realizability and recursive mathematics}, Ph.D. thesis,
  Carnegie-Mellon University, 1984.

\bibitem{Miq07}
A.~Miquel, \emph{Classical program extraction in the calculus of
  constructions}, Computer Science Logic, 21st International Workshop, CSL
  2007, 16th Annual Conference of the EACSL, Lausanne, Switzerland, September
  11-15, 2007, Proceedings, Lecture Notes in Computer Science, vol. 4646,
  Springer, 2007, pp.~313--327.

\bibitem{Miq10}
\bysame, \emph{Existential witness extraction in classical realizability and
  via a negative translation}, Logical Methods for Computer Science (2010).

\bibitem{Forcing}
\bysame, \emph{Forcing as a program transformation}, LICS, IEEE Computer
  Society, 2011, pp.~197--206.

\bibitem{Myh73}
J.~Myhill, \emph{Some properties of intuitionistic {Z}ermelo-{F}raenkel set
  theory}, Lecture Notes in Mathematics \textbf{337} (1973), 206--231.

\bibitem{Oli08}
P.~Oliva and T.~Streicher, \emph{On {K}rivine's realizability interpretation of
  classical second-order arithmetic}, Fundam. Inform. \textbf{84} (2008),
  no.~2, 207--220.

\bibitem{Par97}
M.~Parigot, \emph{Proofs of strong normalisation for second order classical
  natural deduction}, J. Symb. Log. \textbf{62} (1997), no.~4, 1461--1479.

\bibitem{RiegPhD}
Lionel Rieg, \emph{{On Forcing and Classical Realizability}}, Theses, {Ecole
  normale sup{\'e}rieure de lyon - ENS LYON}, June 2014.

\end{thebibliography}
\end{document}